\documentclass[12pt,english]{IEEEtran}
\usepackage[T1]{fontenc}
\usepackage[latin9]{inputenc}
\usepackage{geometry}
\usepackage{color}
\usepackage{float}
\usepackage{textcomp}
\usepackage{amsmath}
\usepackage{amssymb}
\usepackage{stackrel}
\usepackage{graphicx}
\usepackage{esint}
\usepackage{cite}
\usepackage{hyperref} 
\newtheorem{theorem}{\bfseries Theorem}

\makeatletter

\newcommand*\LyXZeroWidthSpace{\hspace{0pt}}
\providecommand{\tabularnewline}{\\}
\floatstyle{ruled}
\newfloat{algorithm}{tbp}{loa}
\providecommand{\algorithmname}{Algorithm}
\floatname{algorithm}{\protect\algorithmname}

\@ifundefined{showcaptionsetup}{}{%
 \PassOptionsToPackage{caption=false}{subfig}}
\usepackage{subfig}
\makeatother

\usepackage{babel}
\begin{document}

\title{RSMA-Enabled ISAC Networks with Fluid Antenna Systems: Stochastic
Geometry Analysis and Low-Complexity Resource Allocation}
\author{\small{Abdelhamid Salem, \textit{Member,~IEEE}, 
            Hana Shamata,
            Salma Elkawafi,
            Khaled M. Rabie, \textit{Senior Member, IEEE},
            Xingwang Li, \textit{Senior Member, IEEE},
            Turki Essa Alharbi, and 
            Mohammed S. Alzaidi}

\vspace{-6mm}

\thanks{A. Salem, Hana Shamata and  Salma Elkawafi are with the Department of Electronic and Electrical Engineering, University of Benghazi, Benghazi, Libya (e-mail: \{abdelhamid.albaraesi,hana.shamata, salma. elkawafi\}@uob.edu.ly).} 
\thanks{Khaled M. Rabie is with the Department of Computer Engineering, King Fahd University of Petroleum and Minerals (KFUPM), Dhahran, Saudi Arabia (email: k.rabie@kfupm.edu.sa).}
\thanks{Xingwang Li is with the School of Physics and Electronic Information
Engineering, Henan Polytechnic University, Jiaozuo 454003, China, and also
with the National Mobile Communications Research Laboratory, Southeast
University, Nanjing 210096, China (e-mail: lixingwangbupt@gmail.com).}
\thanks{Turki Essa Alharbi and Mohammed S. Alzaidi are with the Department of Electrical Engineering, College of Engineering, Taif University, Taif, Saudi Arabia (e-mail: \{turki.alharbi,m.alzaidi\}@tu.edu.sa).}
}

\maketitle
\begin{abstract}
Integrated sensing and communications (ISAC) has recently emerged
as a key technology for sixth-generation (6G) wireless networks by
enabling simultaneous communication and sensing using shared radio
resources. Meanwhile, rate-splitting multiple access (RSMA) offers
robust interference management and improved spectral efficiency, whereas
fluid antenna systems (FAS) provide additional spatial diversity through
low-complexity port switching. Despite these advantages, the joint
performance of RSMA-enabled ISAC networks with FAS remains largely
unexplored, particularly from a network-level perspective. In this
paper, we investigate the downlink performance of multi-cell RSMA-enabled
ISAC networks in which base stations (BSs), communication users, and
sensing targets are spatially distributed according to independent
Poisson point processes (PPPs). Each BS simultaneously serves multiple
users using RSMA while exploiting the common stream as a dual-functional
communication and sensing waveform. The users are equipped with FAS
that selects the best antenna port to maximize the received signal
quality. Closed-form analytical expressions are derived for the ergodic
sum-rates by combining stochastic geometry, order statistics, and
Laplace-transform-based interference analysis. Furthermore, a tractable
approximation for the average radar SINR is developed by characterizing
the statistical properties of the common precoder. Leveraging the
derived analytical expressions, a low-complexity analytical resource
allocation framework is proposed to jointly optimize the RSMA power
allocation, the communication-sensing beam tradeoff, and the number
of scheduled users while satisfying the sensing quality-of-service
constraint. Compared with conventional iterative optimization approaches,
the proposed analytical design significantly reduces computational
complexity while achieving nearly identical communication performance.
Simulation results verify the accuracy of the developed analytical
expressions and demonstrate substantial improvements in both RSMA
sum-rate and sensing performance over conventional transmission schemes.
\end{abstract}


\section{Introduction}

Integrated sensing and communication (ISAC) has emerged as a key enabling technology for sixth-generation (6G) wireless networks by allowing sensing and communication functions to share spectrum, waveforms, hardware platforms, and radio resources \cite{8999605,9540344,9705498,9737357}. This integration can improve spectral, energy, and hardware efficiency by supporting data transmission while simultaneously extracting environmental and target information. Most existing studies have investigated ISAC at the link or system level, typically considering a single base station (BS). Representative contributions include joint radar--communication waveform design \cite{8386661}, multiuser MIMO beamforming \cite{9124713}, optimal transmit beamforming \cite{10086626}, fundamental downlink and uplink performance analysis \cite{9800940}, performance analysis and power allocation for massive MIMO ISAC \cite{10938928}, and joint pilot and transmission design for channel estimation and target detection \cite{10630588}. Although these studies provide important insights into waveform, beamforming, and resource design, their deterministic single-cell models do not explicitly capture spatially random deployments, aggregate inter-cell interference, or cooperation among multiple ISAC transceivers.

Network-level ISAC extends sensing and communication services across geographically distributed BSs, thereby providing broader sensing coverage, spatially diverse observations, and enhanced communication performance \cite{11184506,10726912}. These gains, however, require addressing aggregate inter-cell interference, user and target scheduling, joint resource allocation, and the signaling overhead associated with information exchange among transceivers. Multi-cell precoding under coordinated beamforming, coordinated multipoint transmission, and bistatic sensing was investigated in \cite{10577579}. Stochastic geometry was subsequently employed to analyze inter-cell interference and BS coordination \cite{10735119}, network densification \cite{10490156}, cooperative communication and multistatic sensing under backhaul constraints \cite{10769538}, and antenna topologies ranging from concentrated massive MIMO to distributed cell-free deployments \cite{11030947}. Nevertheless, these network-level studies mainly rely on conventional transmission and fixed antenna architectures.

Rate-splitting multiple access (RSMA) has emerged as a flexible multiple-access and interference-management framework for multi-antenna wireless networks \cite{9831440,10038476,7470942}. Unlike space-division multiple access, which generally treats residual interference as noise, and power-domain non-orthogonal multiple access, which relies on the complete decoding of selected interfering streams according to an SIC order, RSMA enables each receiver to decode part of the interference while treating the remainder as noise \cite{9451194}. In one-layer RSMA, each user message is divided into common and private parts, which are encoded into a common stream and user-specific private streams, respectively. Each user first decodes and removes the common stream before decoding its intended private stream. Through adaptive message splitting, power allocation, and precoding, RSMA provides robustness to imperfect CSIT, heterogeneous channel conditions, and overloaded transmission \cite{7152864,7513415,7555358,7434643,8019852}. Its benefits under practical finite-constellation signaling have also been demonstrated by exploiting, rather than solely suppressing, multiuser interference \cite{9468643}.

RSMA has also been examined in distributed and multi-cell communication networks. Topological rate splitting was developed for interference networks with heterogeneous CSIT qualities \cite{7805217}, while RSMA-enabled user-centric clustering was analyzed in cloud radio access networks \cite{9896157}. A stochastic-geometry framework was further developed for dense multi-cell networks to quantify the effects of BS density, network loading, power allocation, and aggregate interference \cite{10172063}. These studies, however, address communication-only systems. The integration of RSMA with ISAC was initiated through joint common- and private-stream precoding for communication-rate and radar-beampattern design \cite{9531484,9832622}. Subsequent contributions considered multiuser communication and multi-target sensing \cite{10486996}, energy-efficient hardware-constrained architectures \cite{9797869}, transmission under partial CSIT \cite{Loli2022RSMAISAC}, and cell-free massive MIMO ISAC \cite{11192472}. Experimental results have also demonstrated the practical communication and sensing benefits of RSMA \cite{11358849}. Nevertheless, most existing RSMA--ISAC studies adopt deterministic link-level models and fixed antenna configurations.

Complementary to the message-domain interference management offered by RSMA, fluid antenna systems (FAS) provide spatial-domain adaptation by selecting one or more candidate ports within a compact physical aperture \cite{9131873,9264694,10753482}. By exploiting local spatial channel variations, FAS can improve diversity and outage performance and mitigate interference without requiring a proportional increase in active radio-frequency chains. Their implementation principles and electromagnetic reconfigurability were discussed in \cite{10909643}. Since the channels observed at closely spaced ports are spatially correlated, the achievable FAS gain depends on the aperture size, port density, correlation structure, and selection strategy. Correlated-channel approximations, outage and diversity analyses, block-correlation models, and port-selection methods were developed in \cite{10103838,10130117,10623405,9715064}.

FAS have subsequently been extended to multiuser and multiple-access systems \cite{9650760,10794752,refport}. Their integration with RSMA combines spatial channel reconfiguration with partial interference decoding. Existing studies have analyzed the outage and ergodic-rate performance of FAS-assisted RSMA under spatially correlated channels \cite{11134603,11455559}, while broader design opportunities were discussed in \cite{11506060}. These works, however, consider communication-only systems. FAS have also been incorporated into ISAC through joint antenna-position or port-selection and beamforming design \cite{10707252,10705114}. Related studies addressed multiuser FAS--ISAC using deep reinforcement learning \cite{10477314}, sensing-SCNR maximization \cite{10949741}, multi-target multistatic sensing \cite{11175218}, and secure transmission through joint precoding and port selection \cite{11563569}. These contributions predominantly rely on deterministic single-cell models.

A recent preprint jointly considered FAS, RSMA, and ISAC through robust beamforming and antenna-position optimization under communication, sensing, and physical-layer-security constraints \cite{Zhang2025FASRSMAISAC}. However, it focuses on secure link-level design under a deterministic network configuration. To the best of our knowledge, the network-level interplay among RSMA transmission, port selection over spatially correlated FAS channels, random BS--user--target deployments, and aggregate inter-cell interference has not yet been analytically characterized. In particular, a unified stochastic-geometry framework for jointly evaluating common- and private-stream rates and sensing performance in multicell FAS-assisted RSMA--ISAC networks remains unavailable.

Accordingly, this paper develops an analytical framework for FAS-assisted RSMA--ISAC cellular networks. The RSMA common stream is shaped to simultaneously convey common information to the scheduled users and illuminate the sensing target, thereby avoiding a separate sensing waveform. The BSs, users, and sensing targets are modeled as independent PPPs, while port selection over spatially correlated FAS channels and aggregate inter-cell interference are explicitly incorporated into the analysis. The resulting expressions are then used to design a low-complexity resource-allocation framework for the communication--sensing beam tradeoff, common--private power split, and number of scheduled users.

The main contributions are summarized as follows:
\begin{itemize}    
    \item A novel analytical network-level stochastic-geometry framework is developed for multicell FAS-assisted RSMA--ISAC networks. The proposed model jointly captures random BS, user, and target locations, common and private stream transmission, spatially correlated FAS channels with best-port selection, and aggregate inter-cell interference.

    \item New analytical expressions are derived for the communication performance. Tight upper and lower bounds on the ergodic common and private rates are derived using Laplace transforms, order statistics, moment matching, and stochastic-geometry tools. The results explicitly reveal the effects of FAS selection, user ordering, RSMA parameters, and network interference.
    
    \item A tractable approximation of the average sensing SINR is derived by characterizing the statistical properties of the common stream beamformer through moment matching. The result reveals the effects of sensing power, the communication--sensing beam tradeoff, and the spatial distribution of sensing targets.
    
    \item A low-complexity resource-allocation framework is developed to optimize the communication--sensing beamforming parameter, RSMA power split, and number of scheduled users using closed-form solutions and one-dimensional searches. Compared with conventional iterative optimization algorithms, the proposed framework substantially reduces computational complexity while maintaining nearly identical communication performance.

	\item Extensive Monte Carlo simulations validate the analytical results and demonstrate tight analytical accuracy and improved communication--sensing performance with substantially lower complexity than conventional iterative optimization methods. 
\end{itemize}
The remainder of this paper is organized as follows.
Section~II presents the system model, Section~III develops the communication performance analysis and derives
the ergodic common and private rate expressions, while
Section~IV characterizes the sensing performance in terms of the average sensing SINR. Section~V develops the joint RSMA--ISAC resource-allocation framework and presents the proposed low-complexity solution. Section~VI presents
the numerical results and validates the analytical framework through Monte Carlo simulations. Finally, Section~VII concludes the paper.

\section{System Model\label{sec:System-Model}}
We consider a downlink multi-cell RSMA--ISAC network in which distributed BSs simultaneously serve communication
users and sense targets over the same time--frequency resources.
The BS locations are modeled by an independent homogeneous Poisson point process (PPP) $\Phi_b\subset\mathbb{R}^{2}$ with density $\lambda_b$.
Communication users and sensing targets are independently distributed
according to homogeneous PPPs $\Phi_u$ and $\Phi_r$ with densities
$\lambda_u$ and $\lambda_r$, respectively. Each sensing target is located
at altitude $h_t$ above the horizontal network plane. During each
scheduling interval, every BS selects the nearest $K$ users for simultaneous RSMA transmission, where the optimal $K$ will be investigated in this work
\footnote{The number of scheduled users $K$ is treated as a finite design variable rather than a random PPP-dependent quantity. As $K$ increases, the probability of including users with large serving distances or unfavorable channel conditions also increases, and as $K\to\infty$, the common rate $\min(R_{c,k})\to0$ may approach zero and diminish the RSMA gain. Hence, $K$ is restricted to a finite range and optimized later in the paper. This assumption is adopted for analytical tractability and follows commonly used SG-based multiuser network models \cite{10172063}.}.

The proposed model is based on the following assumptions:
\begin{itemize}
    \item Each BS is equipped with $N_s$ fixed transmit antennas for
    downlink RSMA--ISAC transmission and an $N_r$ fixed receive
    antenna for target sensing.
    
    \item All BSs are active and reuse the same bandwidth, causing aggregate inter-cell interference at the communication users and sensing receivers.
    
    \item Each user is equipped with a FAS with $P$ ports uniformly distributed along a linear aperture of length $W\lambda$, where $W$ is a scaling factor and $\lambda$ is the carrier wavelength. Only one port is activated during each transmission interval.
    
    \item No dedicated sensing waveform is transmitted. Instead, the
    RSMA common stream is employed as a dual-functional waveform for
    conveying common information and illuminating the sensing target.
\end{itemize}

\subsection{Channel Model}

To characterize the communication channels of the FAS-equipped users, we adopt a spatially correlated Rayleigh fading model suitable to characterize isotropic scattering environments for compact antenna surfaces \cite{10103838,10623405}.

\subsubsection{BS-to-User Communication Channel Model}
\label{subsubsec:communication_channel}

For small-scale fading between the $n$-th transmit antenna of the typical
BS and the $p$-th FAS port at $k$ user, the complex channel can be modeled as 

\begin{equation}
	\begin{split}
		h_{k,p,n}
		={}&\sigma\left(
		\sqrt{1-\mu^{2}}\,x_{k,p,n}
		+\mu x_{k,1,n}
		\right)\\
		&+j\sigma\left(
		\sqrt{1-\mu^{2}}\,y_{k,p,n}
		+\mu y_{k,1,n}
		\right).
	\end{split}
	\label{eq:fas_channel}
\end{equation}
where $x_{k,p,n}$ and $y_{k,p,n}$ are mutually independent real-valued Gaussian
random variables with zero mean and variance of $1/2$. Consequently, $h_{k,p,n}$
follows a complex normal distribution as $CN(0,\sigma^{2})$ while
the magnitude $\left|h_{k,p,n}\right|$ follows a Rayleigh distribution.
The components $x_{k,1,n}$ and $y_{k,1,n}$ represent reference port variables, i.e., $p=1$, at user $k$. Further, we can define $h_{k,1,n}\triangleq\sigma\left(x_{k,1,n}+jy_{k,1,n}\right)$ as the common random variable that links the random channels at the ports with the correlation parameter that can be chosen as 
\begin{equation}
\mu^{2}
=
\left|
\frac{2}{P(P-1)}
\sum_{q=1}^{P-1}
(P-q)
J_{0}\left(
\frac{2\pi qW}{P-1}
\right)
\right|,
\label{eq:fas_correlation_parameter}
\end{equation}
where $q$ denotes the port separation index and $J_{0}(\cdot)$ is
the zeroth-order Bessel function of the first kind. 

\subsubsection{Sensing Channel Model}
\label{subsubsec:sensing_channel}

For target $m$, the channel between the $N_s$ transmit and the $N_r$ receive array of
the typical BS $o$ is modeled as
\begin{equation}
\mathbf{G}_{m,o}
=
\sqrt{\beta_{m,o}\zeta_{m,o}}\,
\mathbf{a}_{r}(\theta_{m})
\mathbf{a}_{t}^{H}(\theta_{m})
\in
\mathbb{C}^{N_r\times N_s}.
\label{eq:sensing_channel}
\end{equation}
Here, $\beta_{m,o}$ denotes the target distance-dependent sensing power loss, $\zeta_{m,o}$ is the radar
cross section (RCS) of the target, and $\theta_m$ denotes the target direction relative to the typical BS. The transmit steering vector is defined as
\begin{equation}
\mathbf{a}_{t}(\theta_m)
=
\left[
e^{j\frac{2\pi}{\lambda}\triangle_{t}^{1}(\theta_m)},
\ldots,
e^{j\frac{2\pi}{\lambda}\triangle_{t}^{N_s}(\theta_m)}
\right]^{T}
\in
\mathbb{C}^{N_s\times 1},
\label{eq:transmit_steering_vector}
\end{equation}

whereas the receive steering vector is
\begin{equation}
\mathbf{a}_{r}(\theta_m)
=
\left[
e^{j\frac{2\pi}{\lambda}\triangle_{r}^{1}(\theta_m)},
\ldots,
e^{j\frac{2\pi}{\lambda}\triangle_{r}^{N_r}(\theta_m)}
\right]^{T}
\in
\mathbb{C}^{N_r\times 1}.
\label{eq:receive_steering_vector}
\end{equation}
where $\triangle_{m}^{n}(\theta_m)$ denotes the propagation distance difference
between the $n$-th port and the reference port in the direction $\theta_m$.

\subsection{Communication System}

In this work, each BS simultaneously serves $K$ scheduled users using RSMA. At BS $l$, the message intended for user $k$ is divided into a common part and a private part. The common parts of the $K$ user messages are jointly encoded into a common stream $s_{l,c}$, whereas the remaining parts are independently encoded into the user-specific private streams $\{s_{l,i}\}_{i=1}^{K}$. The transmitted symbol
vector at BS $l$ is defined as
\begin{equation}
	\mathbf{s}_{l}
	=
	\left[
	s_{l,c},
	s_{l,1},
	\ldots,
	s_{l,K}
	\right]^{T},
	\qquad
	\mathbb{E}
	\left[
	\mathbf{s}_{l}\mathbf{s}_{l}^{H}
	\right]
	=
	\mathbf{I}_{K+1}.
	\label{eq:rsma_symbol_vector}
\end{equation}

The common and private streams are linearly precoded and superimposed
before transmission. Accordingly, the signal transmitted by the typical BS $o$
is
\begin{equation}
	\begin{aligned}
		\mathbf{x}_{o}
		=
		\mathbf{W}_{o}\mathbf{s}_{o}
		={}&
		\sqrt{P_{o,c}}\,
		\mathbf{w}_{o,c}s_{o,c}
		\\
		&+
		\sum_{i=1}^{K}
		\sqrt{P_{o,i}}\,
		\mathbf{w}_{o,i}s_{o,i},
	\end{aligned}
	\label{eq:transmit_signal}
\end{equation}
where $\mathbf{w}_{o,c}\in\mathbb{C}^{N_s\times1}$ and
$\mathbf{w}_{o,i}\in\mathbb{C}^{N_s\times1}$ are the beamforming vectors of common and private streams, respectively. The power-weighted
precoding matrix is
$\mathbf{W}_{o}
=
[\sqrt{P_{o,c}}\mathbf{w}_{o,c},
\sqrt{P_{o,1}}\mathbf{w}_{o,1},
\ldots,
\sqrt{P_{o,K}}\mathbf{w}_{o,K}]
\in\mathbb{C}^{N_s\times(K+1)}$.
The beamformers and transmit powers satisfy
	$\left\|
	\mathbf{w}_{o,c}
	\right\|^{2}
	=
	\left\|
	\mathbf{w}_{o,i}
	\right\|^{2}
	=
	1,
	P_{o,c}
	+
	\sum_{i=1}^{K}P_{o,i}
	\leq
	P_o,$
where $P_o$ is the total transmit power of BS $o$. 

Unlike conventional ISAC systems that employ a dedicated sensing waveform, the common stream is used as a dual-functional waveform in the proposed system. Its beamformer is jointly designed to convey common information to the scheduled users and illuminate the sensing target. The detailed common- and private-stream precoder designs are
presented in Section~\ref{sec:communication_performance}.


The signal received at the $p$-th port of user $k$ served by the typical BS is
\begin{equation}
	\begin{aligned}
		y_{k,p}
		={}&
		\sqrt{d_{k,o}^{-\alpha_c}P_{o,c}}\,
		\mathbf{h}_{k,p,o}
		\mathbf{w}_{o,c}s_{o,c}
		\\
		&+
		\sum_{i=1}^{K}
		\sqrt{d_{k,o}^{-\alpha_c}P_{o,i}}\,
		\mathbf{h}_{k,p,o}
		\mathbf{w}_{o,i}s_{o,i}
		+ I_{k,p}
		+ n_k,
	\end{aligned}
	\label{eq:received_signal}
\end{equation}
where $\mathbf{h}_{k,p,o}\in\mathbb{C}^{1\times N_{s}}$
is the channel from typical BS to user $k$ at port $p$ $\mathbf{h}_{k,p,o}=\left[h_{k,p,1},.....,h_{k,p,N_{s}}\right]$, $d_{k,o}$ denotes the distance between user $k$ and the typical
BS, $\alpha_c$ is the communication path-loss exponent, $I_{k,p}=\underset{l\in\Phi_{b}\setminus o}{\sum}\sqrt{d_{k,l}^{-\alpha_c}}\mathbf{h}_{k,p,l}\mathbf{x}_{l}$ is the inter-cell interference, and $n_{k}\sim\mathcal{CN}\left(\text{0, }\sigma_{k}^{2}\right)$ is the additive \textit{\emph{white}} Gaussian noise (AWGN) at the user.

At each user, the common stream is decoded first while all private
streams and inter-cell transmissions are treated as interference.
Thus, the common-stream SINR at the $p$-th port of user $k$ is
\begin{equation}
	\gamma_{k,p}^{c}
	=
	\frac{
		d_{k,o}^{-\alpha_c}P_{o,c}
		\left|
		\mathbf{h}_{k,p,o}
		\mathbf{w}_{o,c}
		\right|^{2}
	}{
		\sum_{i=1}^{K}
		d_{k,o}^{-\alpha_c}P_{o,i}
		\left|
		\mathbf{h}_{k,p,o}
		\mathbf{w}_{o,i}
		\right|^{2}
		+
		\underset{l\in\Phi_{b}\setminus o}{\sum}\eta_{k,p,l}
		+
		\sigma_k^{2}
	}.
	\label{eq:common_stream_sinr}
\end{equation}
where $\eta_{k,p,l}=
d_{k,l}^{-\alpha_c}
\left\|
\mathbf{h}_{k,p,l}
\mathbf{W}_{l}
\right\|^{2}$ is the interference power generated by BS $l$ at port $p$ of user $k$.

After successfully decoding the common stream, user $k$ removes it through ideal successive interference cancellation (SIC). The resulting private-stream SINR is
\begin{equation}
	\gamma_{k,p}^{p}
	=
	\frac{
		d_{k,o}^{-\alpha_c}P_{o,k}
		\left|
		\mathbf{h}_{k,p,o}
		\mathbf{w}_{o,k}
		\right|^{2}
	}{
		\sum_{i\neq k}^{K}
		d_{k,o}^{-\alpha_c}P_{o,i}
		\left|
		\mathbf{h}_{k,p,o}
		\mathbf{w}_{o,i}
		\right|^{2}
		+
		\underset{l\in\Phi_{b}\setminus o}{\sum}\eta_{k,p,l}
		+
		\sigma_k^{2}
	}.
	\label{eq:private_stream_sinr}
\end{equation}

Therefore, the achievable common and private stream rates of user $k$ are
\begin{equation}
	R_k^{c}
	=
	\log_2\left(
	1+\gamma_{k,p}^{c}
	\right),
	\qquad
	R_k^{p}
	=
	\log_2\left(
	1+\gamma_{k,p}^{p}
	\right).
	\label{eq:individual_rsma_rates}
\end{equation}

Since the common stream must be decoded by all scheduled users, its achievable rate is limited by the user with the lowest common rate:
\begin{equation}
	R_c
	=
	\min_{1\leq k\leq K}
	R_k^{c}.
	\label{eq:common_rate}
\end{equation}

The instantaneous RSMA sum rate of the typical cell is consequently
given by
\begin{equation}
	R_{\mathrm{sum}}^{\mathrm{RSMA}}
	=
	R_c
	+
	\sum_{i=1}^{K}
	R_i^{p}.
	\label{eq:rsma_sum_rate}
\end{equation}

\LyXZeroWidthSpace{}

\subsection{Sensing System}
The reflected echo signal from target $m$ received at the typical
BS is expressed as

\begin{equation}
	\mathbf y_{m,o}
	=
	\mathbf G_{m,o}\mathbf x_o
	+
	\underbrace{
		\sum_{l\in\Phi_b\setminus\{o\}}
		\sqrt{d_{l,o}^{-\alpha_s}}
		\mathbf H_{l,o}\mathbf x_l
	}_{\text{interference from downlink BSs}}
	+
	\mathbf n_b.
	\label{eq:sensing_received_signal}
\end{equation}

where, $d_{l,o}$ is the distance between interfering BS $l$ and the
typical BS, $\alpha_s$ is the sensing path-loss exponent, $\mathbf x_l$ is the signal transmitted by BS $l$ with power $P_l$, $\mathbf{H}_{l,o}\in\mathbb{C}^{N_r\times N_s}$ denotes the channel from BS $l$ to typical BS. and $\mathbf{n}_b
\sim
\mathcal{CN}
\left(
\mathbf{0},
\sigma_b^{2}\mathbf{I}_{N_r}
\right)$ models the aggregate receiver disturbance, including thermal noise, residual interference and clutter \cite{10490156}.

The typical BS applies a sensing receive beamformer
$\mathbf{w}_{r}\in\mathbb{C}^{N_r\times1}$, normalized such that $
	\left\|
	\mathbf{w}_{r}
	\right\|^{2}
	=
	1.$
Since the common stream is employed as the desired sensing waveform,
the receiver performs matched filtering with respect to $s_{o,c}$.
The private stream echoes are assumed to be suppressed by the
matched-filtering operation, while any residual components are
included in $\mathbf{n}_b$. The resulting filtered sensing signal
is
\begin{equation}
	\begin{aligned}
		\widetilde{y}_{m,o}
		={}&
		\sqrt{P_{o,c}}\,
		\mathbf{w}_{r}^{H}
		\mathbf{G}_{m,o}
		\mathbf{w}_{o,c}s_{o,c}
		\\
		&+
		\sum_{l\in\Phi_{b}\setminus\{o\}}
		\sqrt{d_{l,o}^{-\alpha_s}}
		\mathbf{w}_{r}^{H}
		\mathbf{H}_{l,o}
		\mathbf{x}_{l}
		+
		\mathbf{w}_{r}^{H}
		\mathbf{n}_b.
	\end{aligned}
	\label{eq:filtered_sensing_signal}
\end{equation}

Using the sensing channel in \eqref{eq:sensing_channel}, the
instantaneous sensing SINR for target $m$ is
\begin{equation}
	\gamma_{m}^{s}
	=
	\frac{
		P_{o,c}\beta_{m,o}\zeta_{m,o}
		\left|
		\mathbf{w}_{r}^{H}
		\mathbf{a}_{r}(\theta_m)
		\right|^{2}
		\left|
		\mathbf{a}_{t}^{H}(\theta_m)
		\mathbf{w}_{o,c}
		\right|^{2}
	}{
		\sum_{l\in\Phi_{b}\setminus\{o\}}
		d_{l,o}^{-\alpha_s}
		\left|
		\mathbf{w}_{r}^{H}
		\mathbf{H}_{l,o}
		\mathbf{x}_{l}
		\right|^{2}
		+
		\sigma_b^{2}
		\left\|
		\mathbf{w}_{r}
		\right\|^{2}
	}.
	\label{eq:sensing_sinr}
\end{equation}

A target is successfully detected when $\gamma_{m}^{s}\geq\Gamma_{s}$,
where $\Gamma_{s}$ is the sensing SINR threshold.

\section{Communication Performance}
\label{sec:communication_performance}

In this section, we characterize the ergodic common and private
rates of the typical cell using SG tools. Following
\cite{refport}, the typical BS is assumed to acquire instantaneous
channel state information (CSI) only for the reference port $p=1$ of
each scheduled user, i.e.,
$\{\mathbf{h}_{k,1,o}\}_{k=1}^{K}$. The channels of the remaining
candidate ports are not instantaneously known at the BS and are instead
characterized through the spatial-correlation model in
\eqref{eq:fas_channel}. Each user locally probes its candidate ports
and selects one active port based on the adopted port-selection
criterion. Hence, the precoders are designed using the reference-port
CSI, whereas the FAS gain is obtained through user-side port selection.
The aggregated reference-port channel matrix is defined as
$\mathbf{H}_{1,o}
=
\left[
\mathbf{h}_{1,1,o}^{H},
\ldots,
\mathbf{h}_{K,1,o}^{H}
\right]
\in
\mathbb{C}^{N_s\times K}.$


We employ a maximum-ratio transmission (MRT) joint communication and sensing beamformer for the common stream and zero-forcing (ZF) precoding for the private streams. The common-stream MRT beamformer is given by
\begin{equation}
	\mathbf{w}_{o,c}
	=
	\frac{
		\displaystyle
		\alpha
		\sum_{i=1}^{K}
		\mathbf{h}_{i,1,o}^{H}
		+
		\rho\mathbf{a}_{t}(\theta_m)
	}{
		\displaystyle
		\left\|
		\alpha
		\sum_{i=1}^{K}
		\mathbf{h}_{i,1,o}^{H}
		+
		\rho\mathbf{a}_{t}(\theta_m)
		\right\|
	},
	\label{eq:common_precoder}
\end{equation}
where $\rho\in[0,1]$ is the communication--sensing tradeoff parameter,
and $\alpha=\sqrt{1-\rho^{2}}$ is the corresponding communication weighting
coefficient. Increasing $\rho$ therefore increases the
sensing contribution to the common beamformer while reducing
the communication contribution.


For the private streams, the ZF precoding matrix is
defined as
\begin{equation}
	\widetilde{\mathbf{W}}_{o}^{\mathrm{zf}}
	=
	\mathbf{H}_{1,o}
	\left(
	\mathbf{H}_{1,o}^{H}
	\mathbf{H}_{1,o}
	\right)^{-1},
	\label{eq:zf_precoding_matrix}
\end{equation}
and
\begin{equation}
	\mathbf{w}_{o,i}^{\mathrm{zf}}
	=
	\frac{
		\widetilde{\mathbf{w}}_{o,i}^{\mathrm{zf}}
	}{
		\left\|
		\widetilde{\mathbf{w}}_{o,i}^{\mathrm{zf}}
		\right\|
	},
	\qquad
	i=1,\ldots,K.
	\label{eq:normalized_zf_precoder}
\end{equation}

\subsection{FAS Port Selection}
\label{subsec:fas_port_selection}

For user $k$, the private gain at port $p$ is defined as
\begin{equation}
	G_{k,p}^{p}
	=
	\left|
	\mathbf{h}_{k,p,o}
	\mathbf{w}_{o,k}^{\mathrm{zf}}
	\right|^{2}.
	\label{eq:private_effective_gain}
\end{equation}

The user selects the port that maximizes its private gain. Consequently, the maximum selected gain is:
\begin{equation}
	G_{k,\max}^{p}
	=
	\max_{1\leq p\leq P}
	G_{k,p}^{p}.
	\label{eq:maximum_private_gain}
\end{equation}

\subsection{Ergodic Private Rate--Upper Bound}
\label{subsec:private_rate_upper_bound}

Substituting the ZF beamformer \eqref{eq:normalized_zf_precoder} in
\eqref{eq:private_stream_sinr} and evaluating the resulting SINR at
the selected port yields $\gamma_k^{p}$.

The ergodic private rate of user $k$ is defined as
\begin{equation}
	\bar{R}_{k}^{p}
	=
	\mathbb{E}
	\left[
	\log_{2}
	\left(
	1+\gamma_{k}^{p}
	\right)
	\right].
	\label{eq:ergodic_private_rate}
\end{equation}
The exact ergodic private and common rates are analytically intractable due to the coupled effects of FAS port selection, multiuser precoding, and aggregate inter-cell interference. Therefore, tractable upper and lower bounds are derived.

\begin{theorem}
	An upper bound on the ergodic private-stream rate of user $k$ is
	given by
	\begin{equation}
		\bar{R}_{k,\mathrm{Ub}}^{p}
		=
		\frac{1}{\ln 2}
		\sum_{n=1}^{N}
		\frac{H_n}{z_n}
		\left(
		1-x_{k}^{p}
		\right)
		e^{z_n}
		y_{k}^{p},
		\label{eq:private_rate_upper_bound}
	\end{equation}
	\emph{where}
	\begin{equation}
		\begin{aligned}
			x_k^p
			={}&
			\frac{2(\pi\lambda_u)^k z_nP_{o,k}}
			{(k-1)!\Gamma(m)}
			\sum_{r=0}^{P}(-1)^r\binom{P}{r}
			\\
			&\times
			\sum_{q=0}^{r(m-1)}
			c_{r,q}\Gamma(m+q)
			\int_{0}^{\infty}
			\frac{r_o^{2k-1-\alpha_c}e^{-\pi\lambda_u r_o^2}}
			{\left(r_o^{-\alpha_c}P_{o,k}z_n+r\right)^{m+q}}
			\,\mathrm{d}r_o .
		\end{aligned}
		\label{eq:xkp}
	\end{equation}
	\begin{equation}
		\begin{aligned}
			y_{k}^{p}
			={}&
			\frac{
				2(\pi\lambda_{u})^{k}
				e^{-z_n\sigma_{k}^{2}}
			}{
				(k-1)!
			}
			\\
			&\times
			e^{
				-2\pi\lambda_{b}
				\int_{r_{l}}^{\infty}
				\left[
				1-
				\left(
				1+d_{k,l}^{-\alpha_c}z_{n}
				\right)^{-K}
				\right]
				d_{k,l}\,
				\mathrm{d}d_{k,l}
			}
			\\
			&\times
			\int_{0}^{\infty}
			\left(
			1+r_{o}^{-\alpha_c}P_{o,k}z_{n}
			\right)^{1-K}
			r_{o}^{\,2k-1}
			e^{-\pi\lambda_{u}r_{o}^{2}}
			\,\mathrm{d}r_{o}.
		\end{aligned}
		\label{eq:ykp}
	\end{equation}
	\emph{while $z_n$ and $H_n$ are the $n$-th Gauss--Laguerre abscissa and
		weight, respectively, $r_o=d_{k,o}$, and $N$ denotes the quadrature order, $m=N_{s}-K+1$,
		$c_{r,q}=\underset{\underset{0\leq k_{i}\leq m-1}{k_{1}+..+k_{r}=q}}{\sum}\frac{1}{k_{1}!k_{2}!..k_{r}!}$.}
\end{theorem}

\begin{IEEEproof}
The proof is provided in Appendix A.
\end{IEEEproof}

\subsection{Ergodic Private Rate Lower Bound}
\label{subsec:private_rate_lower_bound}

A lower bound on the ergodic private rate is obtained by
applying Jensen's inequality to the logarithmic SINR expression.

\begin{theorem}
	A lower bound on the ergodic private-stream rate of user $k$ is
	denoted by
	\begin{align}
		R_{k,\mathrm{Lb}}^{p}
		={}&
		\frac{1}{\ln 2}
		\sum_{n=1}^{N} H_{n}
		\ln\Biggl(
		1+
		\nonumber\\[-1mm]
		&\qquad
		\frac{
			P_{o,k}r_{o}^{-\alpha_c}
			+
			e^{\psi\left(N_{s}-K+1\right)}
		}{
			P_{o,k}r_o^{-\alpha_c}(K-1)
			+
			\frac{2\pi\lambda_{b}K}{\alpha_c-2}
			r_{o}^{\,2-\alpha_c}
			+
			\sigma_{k}^{2}
		}
		\Biggr)
		\nonumber\\
		&\times
		\frac{2(\pi\lambda_{u})^{k}}{(k-1)!}
		r_{o}^{\,2k-1}
		e^{-\pi\lambda_{u}r_{o}^{2}} .
		\label{eq:private_rate_lower_bound}
	\end{align}
	\emph{where $\psi\left(.\right)$ is the digamma function, }$H_{n}$
	and $r_{n}$ are the $n^{th}$ zero and the weighting factor of the
	Laguerre polynomials, respectively.
\end{theorem}

\begin{IEEEproof}
The proof is provided in Appendix B.
\end{IEEEproof}

\subsection{Ergodic Common Rate-Upper Bound}
\label{subsec:common_rate_upper_bound}

The received SINR of the common part at user $k$ is obtained by substituting (\ref{eq:common_precoder}) and (\ref{eq:zf_precoding_matrix}) in (\ref{eq:common_stream_sinr})

\begin{theorem}
	An upper bound on the ergodic common rate of user $k$ is
	given by
	\begin{equation}
		\bar{R}_{k,\mathrm{Ub}}^{c}
		=
		\frac{1}{\ln 2}
		\sum_{n=1}^{N}
		\frac{H_n}{z_n}
		\left(
		1-x_{k}^{c}
		\right)
		e^{z_n}
		y_{k}^{c}.
		\label{eq:common_rate_upper_bound}
	\end{equation}	
	
	\emph{where}
	
	\begin{align}
		x_{k}^{c}
		={}&
		\frac{2(\pi\lambda_{u})^{k}}{(k-1)!}
		\int_{0}^{\infty}
		\left(
		1+\theta_{c}z_n r_{o}^{-\alpha_c}P_{o,c}
		\right)^{-m_{c}}
		\nonumber\\
		&\times
		r_{o}^{\,2k-1}
		e^{-\pi\lambda_{u}r_{o}^{2}}
		\,\mathrm{d}r_{o},
		\label{eq:xkc}
		\\[1mm]
		y_{k}^{c}
		={}&
		\frac{2(\pi\lambda_{u})^{k}
			e^{-\sigma_{k}^{2}}}{(k-1)!}
		\nonumber\\
		&\times
		e^{-2\pi\lambda_{b}
			\int_{r_{l}}^{\infty}
			\left[
			1-
			\left(
			1+d_{k,l}^{-\alpha_c}z_n
			\right)^{-K}
			\right]
			d_{k,l}\,
			d_{d{k,l}}}
		\nonumber\\
		&\times
		\int_{0}^{\infty}
		\left(
		1+r_{o}^{-\alpha_c}P_{o,k}z_n
		\right)^{-K}
		r_{o}^{\,2k-1}
		e^{-\pi\lambda_{u}r_{o}^{2}}
		\,\mathrm{d}r_{o}
		\label{eq:ykc}
	\end{align}
	\emph{while $z_n$ and $H_n$ are the $n$-th Gauss--Laguerre abscissa and
		weight, respectively, and $N$ denotes the quadrature order. $m_c$ and $\theta_c$ are the shape and scale parameters, respectively, of the moment-matched Gamma distribution $G_c\sim\Gamma(m_c,\theta_c)$, with}
		$m_c
		=
		\frac{\left(\mathbb{E}[G_c]\right)^2}
		{\operatorname{Var}(G_c)},
		\theta_c
		=
		\frac{\operatorname{Var}(G_c)}
		{\mathbb{E}[G_c]},$ and
		$\mathbb{E}[G_c]
		\approx
		1+
		\frac{N_s\alpha^2}
		{K\alpha^2+\rho^2},
		\operatorname{Var}(G_c)
		\approx
		1+
		\frac{2N_s\alpha^2}
		{K\alpha^2+\rho^2}.$
\end{theorem}

\begin{IEEEproof}
The proof is provided in Appendix C.
\end{IEEEproof}


\subsection{Ergodic Common Rate-Lower bound}
\label{subsec:common_rate_lower_bound}

A lower bound on the ergodic common rate is obtained by applying
Jensen's inequality to the moment-matched common gain.
\begin{theorem}
	A lower bound on the ergodic common rate of user $k$ is denoted
	by
	\begin{align}
		R_{k,\mathrm{Lb}}^{c}
		={}&
		\frac{1}{\ln 2}
		\sum_{n=1}^{N}
		H_{n}
		\ln\Biggl(
		1+
		\nonumber\\[-1mm]
		&\qquad
		\frac{
			P_{o,c}r_{o}^{-\alpha_c}\theta_{c}
			e^{\psi\left(m_{c}\right)}
		}{
			KP_{o,k}r_{o}^{-\alpha_c}
			+
			\frac{2\pi\lambda_{b}K}{\alpha_c-2}
			r_{o}^{\,2-\alpha_c}
			+
			\sigma_{k}^{2}
		}
		\Biggr)
		\nonumber\\
		&\times
		\frac{2(\pi\lambda_{u})^{k}}{(k-1)!}
		r_{o}^{\,2k-1}
		e^{-\pi\lambda_{u}r_{o}^{2}} .
		\label{eq:common_rate_lower_bound}
	\end{align}
	
\end{theorem}

\begin{IEEEproof}
The proof is provided in Appendix D.
\end{IEEEproof}

\section{Sensing Performance}
\label{sec:sensing_performance}

Since the RSMA common stream is employed as the sensing waveform, we
characterize the sensing performance through the average sensing SINR
defined from \eqref{eq:sensing_sinr}.

\begin{theorem}
	The average sensing SINR for target $m$ in the considered
	RSMA--ISAC network is approximated by
	\begin{equation}
		\begin{aligned}
			\bar{\gamma}_{s}
			={}&
			\frac{
				(\alpha_{s}-2)\zeta_{m,o}P_{o,c}
				\left|
				\mathbf{w}_{r}^{H}
				\mathbf{a}_{r}(\theta_{m})
				\right|^{2}
			}{
				2\lambda_{b} {P_l} K
			}
			\lambda_{r} (\pi \lambda_{r})^\frac {\alpha_s} 2
			e^{\pi\lambda_{r}h_t^{2}}
			\\
			&\times
			\Gamma
			\left(\frac {-\alpha_s} 2,			\pi\lambda_{r}h_{t}^{2}
			\right)
			\frac{
				A_{\alpha}
				+
				\rho^{2}N_{s}
			}{
				A_{\alpha}
				+
				\rho^{2}
			}.
		\end{aligned}
		\label{eq:average_sensing_sinr}
	\end{equation}
\end{theorem}

where  $A_{\alpha}=K(1-\rho^{2})$ is the communication component of the common beamformer.

\begin{IEEEproof}
	The proof is provided in Appendix E.
\end{IEEEproof}

\section{Joint RSMA\textendash ISAC Resource Allocation}
\label{sec:resource_allocation}

The analytical expressions derived in the previous sections characterize
the communication and sensing performance of the considered
RSMA--ISAC network in terms of the system parameters.
Based on these expressions, we formulate a joint resource-allocation
problem to balance the communication and sensing performance.
Specifically, we optimize the communication\textendash sensing beam tradeoff by beamforming coefficient $\rho$, the RSMA power allocation, and the number of scheduled users $K$ to maximize the average RSMA sum rate while satisfying the sensing quality-of-service (QoS) requirement. The total transmit power of the typical BS is divided between the common and private streams according to 
	$P_{o,c}
	=
	\beta P_o,
	P_{o,k}
	=
	\frac{(1-\beta)P_o}{K},$ where
	$0\leq\beta\leq1.$
Accordingly, the average RSMA sum rate is expressed as
\begin{equation}
	\bar{R}_{\mathrm{sum}}^{\mathrm{RSMA}}
	(\rho,\beta,K)
	=
	\bar{R}_{c}(\rho,\beta,K)
	+
	\sum_{k=1}^{K}
	\bar{R}_{k}^{p}(\beta,K),
	\label{eq:average_rsma_sum_rate}
\end{equation}

The joint resource-allocation problem is formulated as
\begin{equation}
	\begin{aligned}
		(\mathcal{P}_{1}):\qquad
		\underset{\rho,\beta,K}{\operatorname{max}}
		\quad&
		\bar{R}_{\mathrm{sum}}^{\mathrm{RSMA}}
		(\rho,\beta,K)
		\\
		\textrm{s.t.}
		\quad&
		\bar{\gamma}_{s}(\rho,\beta,K)
		\geq
		\Gamma_s,
		\\
		&
		0\leq\rho\leq1,
		\\
		&
		0\leq\beta\leq1,
		\\
		&
		K_{\min}\leq K\leq K_{\max},
		\qquad
		K\in\mathbb{Z}^{+},
	\end{aligned}
	\label{eq:joint_resource_allocation}
\end{equation}
where $\Gamma_s$ denotes the minimum required average sensing SINR, $K_{\min}\geq1$, and $K_{\max}$ is the maximum number
of scheduled users. Since ZF precoding is employed
for the private streams, the scheduling size must satisfy
$K_{\max}\leq N_s,$ Problem $(\mathcal{P}_{1})$ is a mixed integer non-convex optimization problem because of the discrete scheduling variable $K$, the nonlinear coupling among $\rho$, $\beta$, and $K$, and the sensing QoS
constraint.

\subsection{Proposed Block Coordinate Descent Algorithm}

Although $(\mathcal{P}_{1})$ is non-convex, its optimization variables can be separated into two continuous variables, $\rho$ and $\beta$, and one discrete variable, $K$. Accordingly, a block coordinate descent (BCD) algorithm is adopted, in which one variable is optimized
while the remaining variables are kept fixed.
At the $t$-th iteration, the joint problem is decomposed into three subproblems.

\subsubsection{Optimization of the communication\textendash sensing beam tradeoff}
For fixed $\beta^{(t)}$ and $K^{(t)}$, the communication--sensing
tradeoff parameter is updated as
\begin{equation}
	\begin{aligned}
		\rho^{(t+1)}
		=
		\underset{0\leq\rho\leq1}{\operatorname{arg\,max}}
		\quad&
		\bar{R}_{\mathrm{sum}}^{\mathrm{RSMA}}
		\left(
		\rho,\beta^{(t)},K^{(t)}
		\right)
		\\
		\textrm{s.t.}
		\quad&
		\bar{\gamma}_{s}
		\left(
		\rho,\beta^{(t)},K^{(t)}
		\right)
		\geq
		\Gamma_s.
	\end{aligned}
	\label{eq:rho_subproblem}
\end{equation}

Since \eqref{eq:rho_subproblem} involves only one scalar variable that is optimized over
the bounded interval $[0,1]$, it can be efficiently solved using a
one-dimensional search such as the Golden Section Search (GSS) which is stated in Algorithm 1.

\subsubsection{Optimization of RSMA power allocation}
For fixed $\rho^{(t+1)}$ and $K^{(t)}$, the power-allocation
coefficient is updated according to
\begin{equation}
	\begin{aligned}
		\beta^{(t+1)}
		=
		\underset{0\leq\beta\leq1}{\operatorname{arg\,max}}
		\quad&
		\bar{R}_{\mathrm{sum}}^{\mathrm{RSMA}}
		\left(
		\rho^{(t+1)},\beta,K^{(t)}
		\right)
		\\
		\textrm{s.t.}
		\quad&
		\bar{\gamma}_{s}
		\left(
		\rho^{(t+1)},\beta,K^{(t)}
		\right)
		\geq
		\Gamma_s.
	\end{aligned}
	\label{eq:beta_subproblem}
\end{equation}

Again, only one optimization variable, $\beta$, is involved and can be efficiently determined through a one-dimensional GSS over
$[0,1]$.

\subsubsection{Optimization of the number of scheduled users}
For fixed $\rho^{(t+1)}$ and $\beta^{(t+1)}$, the number of scheduled
users is updated as
\begin{equation}
	K^{(t+1)}
	=
	\underset{
		K_{\min}\le K\le K_{\max}
	}{
		\operatorname{arg\,max}
	}
	\;
	\bar{R}_{\mathrm{sum}}^{\mathrm{RSMA}}
	\left(
	\rho^{(t+1)},
	\beta^{(t+1)},
	K
	\right),
	\label{eq:K_subproblem}
\end{equation}
Since $K$ is an integer variable, the solution is obtained by searching over the finite set $\{K_{\min},K_{\min}+1,\ldots,K_{\max}\}$, as summarized in
Algorithm~2.

The three optimization steps are repeated until the change in the average RSMA
sum rate satisfies
$\left|\bar{R}_{{\rm sum}}^{(t+1)}-\bar{R}_{{\rm sum}}^{(t)}\right|\le\varepsilon,$
where $\varepsilon$ is a prescribed convergence tolerance. The
overall iterative procedure is summarized in Algorithm~3

\begin{algorithm}[H]
	Input: Initial interval $[0,1]$, fixed ($\beta$), fixed ($K$),
	convergence tolerance ($\varepsilon$).
	
	Output: Optimal communication\textendash sensing tradeoff parameter
	($\rho^{\star}$).
	
	1. Initialize $a=0,\qquad b=1,$ and compute the golden ratio coefficient
	$\phi=\frac{\sqrt{5}-1}{2}\approx0.618.$ 
	
	2. Compute the two interior points
	
	$\rho_{1}=b-\phi(b-a),$ $\rho_{2}=a+\phi(b-a).$ 
	
	3. Evaluate the objective function
	
	$F(\rho)=\bar{R}_{{\rm sum}}^{{\rm RSMA}}(\rho,\beta,K),$
	
	while verifying the sensing constraint
	
	$\bar{\gamma}_{s}(\rho,\beta,K)\ge\Gamma_{s}.$ 
	
	If the constraint is violated, assign
	
	$F(\rho)=-\infty.$ 
	
	4. Compare $F(\rho_{1})$ and $F(\rho_{2})$.
	
	If $F(\rho_{1})>F(\rho_{2}),$ update $b=\rho_{2}.$ 
	
	Otherwise, $a=\rho_{1}.$ 
	
	5. Recompute
	
	$\rho_{1}=b-\phi(b-a),$ $\rho_{2}=a+\phi(b-a).$
	
	6. Repeat Steps 3\textendash 5 until $|b-a|<\varepsilon.$ 
	
	7. Return $\rho^{\star}=\frac{a+b}{2}.$
	
	\protect\caption{GSS for Optimizing the Communication\textendash Sensing Tradeoff Parameter
		$\rho$.}
\end{algorithm}
\begin{algorithm}[H]
Input: Fixed ($\rho$), fixed ($\beta$), search interval $K\in\left\{ K_{\min},K_{\min}+1,\ldots,K_{\max}\right\} $. 

Output: Optimal number of scheduled users ($K^{\star}$).

1. Initialize $R_{\max}=-\infty$. 

2. For each $K=K_{\min},K_{\min}+1,\ldots,K_{\max}$, 

perform the following steps:

a) Compute the analytical average common rate $\bar{R}_{c}(\rho,\beta,K).$

b) Compute the analytical average private rate $\sum_{k=1}^{K}\bar{R}_{k}^{p}(\beta,K).$ 

c) Compute the average radar SINR $\bar{\gamma}_{s}(\rho,\beta,K).$

d) If $\bar{\gamma}_{s}<\Gamma_{s},$ discard the current ($K$).

Otherwise, compute $R(K)=\bar{R}_{c}+\sum_{k=1}^{K}\bar{R}_{k}^{p}.$ 

e) If $R(K)>R_{\max}$, update $R_{\max}=R(K),$ 

and $K^{\star}=K.$ 

3. Return the optimal scheduling parameter $K^{\star}.$ 

\protect\caption{One-Dimensional Integer Search for Optimizing the Number of Scheduled
Users ($K$).}
\end{algorithm}
\begin{algorithm}[H]
Input: ($P_{o}$,$\Gamma_{s}$,$K_{\min}$,$K_{\max}$,$\varepsilon$).

1. Initialize $\rho^{(0)}$,$\beta^{(0)}$,$K^{(0)}$. 

2. Compute the analytical common and private rates. 

3. Repeat

$\:$$\:$Update ($\rho$) using Golden Section Search. 

$\:$$\:$Update ($\beta$) using Golden Section Search. 

$\:$$\:$Update ($K$) using one-dimensional integer search. 

4. Until convergence. 

5. Output ($\rho^{\star},\beta^{\star},K^{\star}$).

\protect\caption{Proposed Joint RSMA\textendash ISAC Resource Allocation.}
\end{algorithm}

\subsubsection{Convergence Analysis}
Assuming that each subproblem is solved to its optimal feasible value, the resulting objective sequence is
non-decreasing, i.e.,
	$\bar{R}^{(t)}
	\leq
	\bar{R}_{\rho}^{(t)}
	\leq
	\bar{R}_{\beta}^{(t)}
	\leq
	\bar{R}_{K}^{(t)}
	=
	\bar{R}^{(t+1)}.$

Furthermore, the achievable RSMA sum rate is upper bounded for finite
transmit power. Hence, the non-decreasing sequence
$\{\bar{R}^{(t)}\}$ converges to a finite value. Due to the
non-convex and mixed-integer nature of $(\mathcal{P}_{1})$, the
resulting solution is not necessarily globally optimal.

\subsubsection{Computational Complexity}

Let $L_{\rho}$ and $L_{\beta}$ denote the numbers of GSS iterations
required to optimize $\rho$ and $\beta$, respectively, and let $T$
denote the total number of BCD iterations. The computational complexity of one BCD iteration is
	$\mathcal{O}
	\left(
	L_{\rho}
	+
	L_{\beta}
	+
	(K_{\max}-K_{\min}+1)
	\right).$
Accordingly, the total complexity is
	$\mathcal{O}
	\left(
	T
	\left[
	L_{\rho}
	+
	L_{\beta}
	+
	(K_{\max}-K_{\min}+1)
	\right]
	\right).$
	
For comparison, an exhaustive search using $N_{\rho}$ and $N_{\beta}$
grid points for $\rho$ and $\beta$, respectively, requires
	$\mathcal{O}
	\left(
	N_{\rho}N_{\beta}
	\left(
	K_{\max}-K_{\min}+1
	\right)
	\right).$
Hence, the proposed BCD algorithm substantially reduces the computational complexity compared with a full three-dimensional exhaustive search while providing monotonic convergence. 

\subsection{Proposed Analytical Solution}

Although the BCD approach provides an iterative solution to
$(\mathcal{P}_{1})$, repeatedly evaluating the analytical rate under
sensing constraint may still incur non-negligible computational
cost. Motivated by the derived analytical performance
expressions derived in Sections ~\ref{sec:communication_performance} and \ref{sec:sensing_performance}, we next develop a low-complexity sequential resource allocation method for determining $\beta$, $\rho$, and $K$.

\subsubsection{Communication--Sensing Beam Tradeoff}

The common precoder in \eqref{eq:common_precoder} combines the communication component and the sensing steering vector through the beam tradeoff coefficient $\rho$.
Assuming that the communication objective decreases as additional
weight is assigned to sensing, the minimum feasible value of $\rho$
is obtained by activating the sensing constraint $\bar{\gamma}_{s}=\Gamma_s$.

Solving \eqref{eq:average_sensing_sinr} with respect to $\rho$ yields
the closed-form candidate
\begin{equation}
	\rho^{\star}
	=
	\sqrt{
		\frac{
			P_{o,c}K-zK
		}{
			P_{o,c}K
			+
			P_{o,c}N_s
			+
			zK-z
		}
	},
	\label{eq:rho_closed_form}
\end{equation}
where $z=\frac{\Gamma_{s}}{\varrho}$, and $\varrho=\frac{\left(\alpha_{s}-2\right)\zeta_{m,o}P_{o,c}\left|\mathbf{w}_{r}^{H}\mathbf{a}_{r}\left(\theta_{m}\right)\right|^{2}}{2\lambda_{b} {P_l} K}e^{\pi\lambda_{r}h_{t}^{2}}\\\lambda_{r} (\pi \lambda_{r} )^{\frac{\alpha_s} 2} \Gamma\left(\frac{-\alpha_s} 2,\pi\lambda_{r}h_{t}^{2}\right).$
Thus, the communication\textendash sensing beamforming coefficient
can be computed directly without any iterative optimization.

\subsubsection{RSMA Power Allocation}

The power-allocation coefficient $\beta$ determines the fraction of
the total transmit power assigned to the common stream. To preserve the private-stream
performance relative to the conventional NoRS transmission while
allocating the remaining power to the common stream. The private-stream performance is designed to approximately preserve
the NoRS benchmark:
\begin{equation}
	\mathbb{E}
	\left[
	R_k^p
	\right]
	\approx
	\mathbb{E}
	\left[
	R_k^{\mathrm{NoRS}}
	\right],
	\qquad
	k=1,\ldots,K.
	\label{eq:private_rate_preservation}
\end{equation}

Accordingly, the power-allocation coefficient can be determined from
\begin{equation}
	\beta^{\star}
	=
	\underset{0\leq\beta\leq1}{\operatorname{arg\,min}}
	\left|
	\mathbb{E}
	\left[
	R_k^p
	\right]
	-
	\mathbb{E}
	\left[
	R_k^{\mathrm{NoRS}}
	\right]
	\right|.
	\label{eq:beta_analytical_solution}
\end{equation}


\subsubsection{Number of Users}

In RSMA, the achievable common rate is determined by the
scheduled user with the lowest common rate as in (\ref{eq:common_rate}). Therefore, maximizing the number of simultaneously scheduled users without considering the common rate bottleneck may significantly degrade the overall system performance. To maintain reliable common-stream decoding, we impose
\begin{equation}
	\min_{1\leq k\leq K}
	\bar{R}_{k}^{c}
	\geq
	\varepsilon_c,
	\label{eq:common_rate_qos}
\end{equation}
where $\varepsilon_c$ denotes the minimum required average common rate. Accordingly, the optimal number of scheduled users is determined as
\begin{equation}
	K^{\star}
	=
	\max
	\left\{
	K:
	\min_{1\leq k\leq K}
	\bar{R}_{k}^{c}
	\geq
	\varepsilon_c,
	\;
	K\leq N_s
	\right\}.
	\label{eq:optimal_user_number}
\end{equation}

The proposed scheduling algorithm here starts from the maximum feasible
multiplexing order $K=N_{s}$. For each candidate $K$, the corresponding resource-allocation parameters are determined and the common-rate is evaluated. If the common rate constraint is satisfied, the current scheduling size is accepted. Otherwise, the number of scheduled users is reduced by one $K=N_s-1$ and the procedure is repeated until the feasibility condition is met.

Since the search is performed only over the finite set $\left\{ 1,2,\ldots,N_{s}\right\} ,$
the complexity of the scheduling algorithm is linear with respect
to the number of transmit antennas, i.e., $\mathcal{O}(N_{s})$.

Consequently, the proposed analytical resource-allocation framework
requires only a one-dimensional search for the RSMA power-allocation
coefficient $\beta$, a closed-form evaluation of the communication--sensing
tradeoff parameter $\rho$, and a finite search over the feasible scheduling
set $K$. This substantially reduces the computational burden relative to
conventional iterative optimization algorithms.

\section{Numerical Results}

This section validates the analytical communication and sensing results developed in Sections~\ref{sec:communication_performance}
and~\ref{sec:sensing_performance}, and evaluates the performance and computational efficiency of the resource-allocation framework developed in Section~\ref{sec:resource_allocation}. Monte Carlo
results are averaged over $10^{5}$ independent realizations of the network topology and small-scale fading. Unless otherwise specified, the parameters are summarized in Table~\ref{tab:simulation_parameters}.

\begin{table}[H]
	\centering
	\caption{Simulation parameters.}
	\label{tab:simulation_parameters}
	\begin{tabular}{|c|c|}
		\hline
		Parameter & Value \\
		\hline
		\hline 
		BS density
		& $\lambda_b=5\times10^{-6}\,\mathrm{BS/m^2}$ \\
		\hline
		User density
		& $\lambda_u=5\times10^{-5}\,\mathrm{users/m^2}$ \\
		\hline
		Target density
		& $\lambda_r=2\times10^{-6}\,\mathrm{targets/m^2}$ \\
		\hline
		Communication path-loss exponent
		& $\alpha_c=2.7$ \\
		\hline
		Sensing path-loss exponent
		& $\alpha_s=2.7$ \\
		\hline
		Number of transmit antennas
		& $N_s=8$ \\
		\hline
		Number of receive antennas
		& $N_r=8$ \\
		\hline
		Number of FAS ports
		& $P=16$ \\
		\hline
		Number of scheduled users
		& $K=4$ \\
		\hline
		Noise variance
		& $\sigma_k^2=\sigma_b^2=1$ \\
		\hline
		Sensing SINR threshold
		& $\Gamma_s=-5~\mathrm{dB}$ \\
		\hline
		Target altitude
		& $h_r=20~\mathrm{m}$ \\
		\hline
		Radar cross section
		& $\zeta_{m,o}=1$ \\
		\hline
		Power allocation coefficient
		& $\beta=0.1$ \\
		\hline
		Beamforming coefficient 
		& $\rho=0.5$\tabularnewline
		\hline 
	\end{tabular}
\end{table}

\begin{figure}[t]
	\centering
	\includegraphics[scale=0.6]{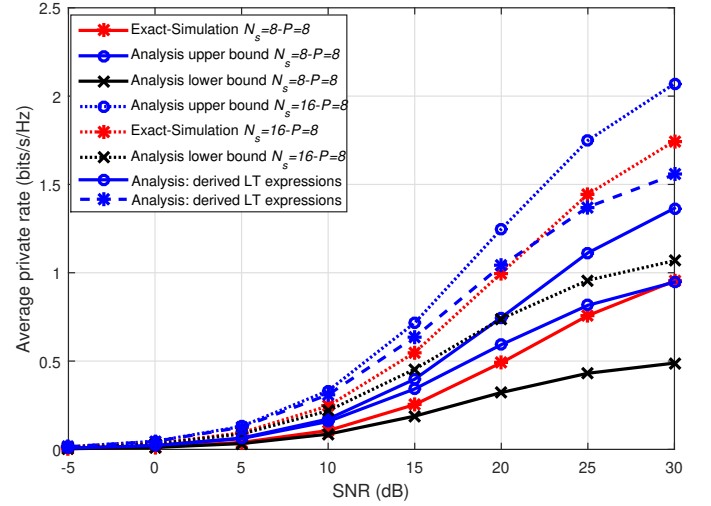}
	\caption{Average private rate versus transmit SNR.}
	\label{fig:private_rate_snr}
\end{figure}

Fig.~\ref{fig:private_rate_snr} illustrates the average private rate versus the transmit SNR for different numbers of transmit antennas and compares the simulated private
rate with the analytical upper and lower bounds. The private rate increases with the transmit SNR for both antenna configurations. Increasing the transmit-array size from $N_s=8$ to $N_s=16$ provides a substantial improvement
over the entire SNR range because the additional spatial degrees of freedom increase the effective ZF beamforming gain and improve multiuser interference suppression. In particular, the ZF gain depends on the dimension $N_s-K+1$, which increases with $N_s$ for a fixed number of scheduled users. The Monte Carlo results are consistently
bounded by the derived analytical upper and lower bounds, validating the proposed communication-performance analysis.

\begin{figure}[t]
	\centering
	\includegraphics[scale=0.6]{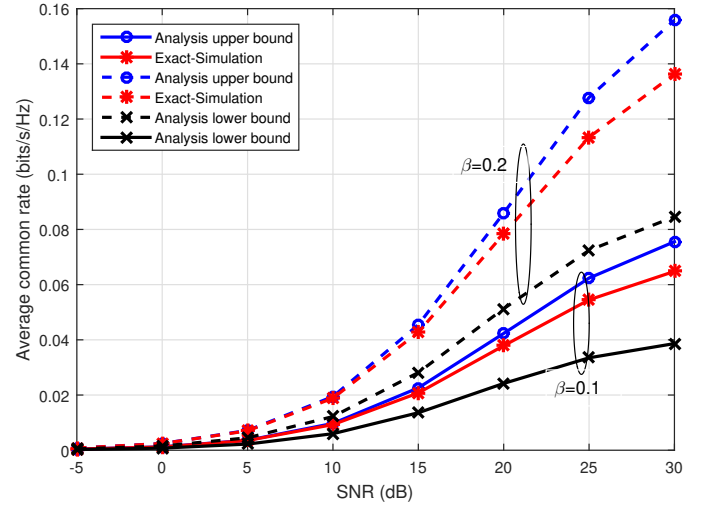}
	\caption{Average common rate versus transmit SNR for different common-stream power allocations with $P=8$.}
	\label{fig:common_rate_snr}
\end{figure}

Fig.~\ref{fig:common_rate_snr} shows the average common rate versus the transmit SNR for two common-stream power allocations with $P=8$. The analytical upper and lower bounds consistently enclose the Monte Carlo results and reproduce the effect of both transmit SNR and common-stream power allocation. The common rate increases with the transmit SNR, while its magnitude remains relatively smaller than the aggregate private rate because the common stream limited by the weakest common-stream decoder. Increasing the power allocation coefficient from $\beta=0.1$ to $\beta=0.2$, improves the common rate over the entire SNR range, demonstrating the impact of RSMA power allocation in RSMA transmission.

\begin{figure}[t]
	\centering
	\includegraphics[scale=0.6]{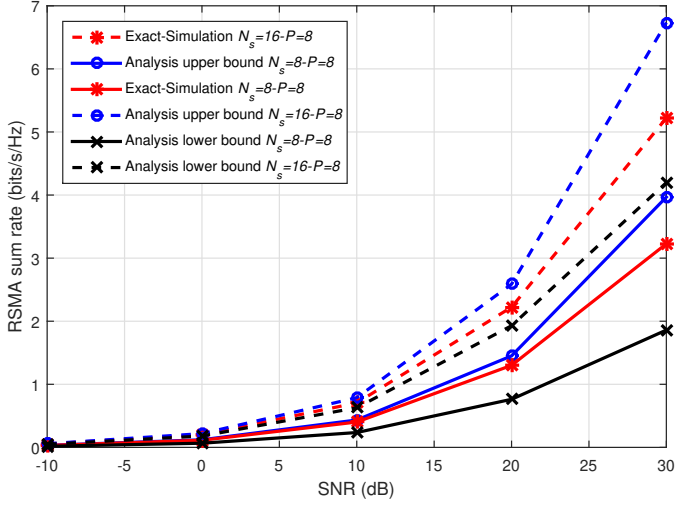}
	\caption{Average RSMA sum rate versus transmit SNR.}
	\label{fig:sum_rate_snr}
\end{figure}

Fig.~\ref{fig:sum_rate_snr} presents the average RSMA sum rate versus the transmit SNR. The Monte Carlo results remain between the analytical bounds, while the upper bound tracks the simulated performance more closely than the lower bound. Increasing $N_s$ from $8$ to $16$ produces a clear rate improvement, particularly at high SNR, due to the additional spatial degrees of freedom and enhanced beamforming gain. This behavior is consistent with the private rate results in Fig.~\ref{fig:private_rate_snr}, since the private-stream contribution constitutes the dominant component of the total RSMA rate for the considered parameters.
\begin{figure}[t]
	\centering
	\subfloat[Sum of the average private-stream rates versus $K$.]{
		\includegraphics[scale=0.5]{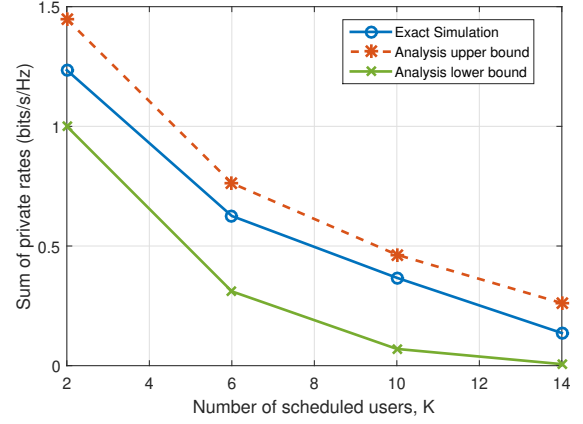}
	}
	\\
	\subfloat[Average common-stream rate versus $K$.]{
		\includegraphics[scale=0.5]{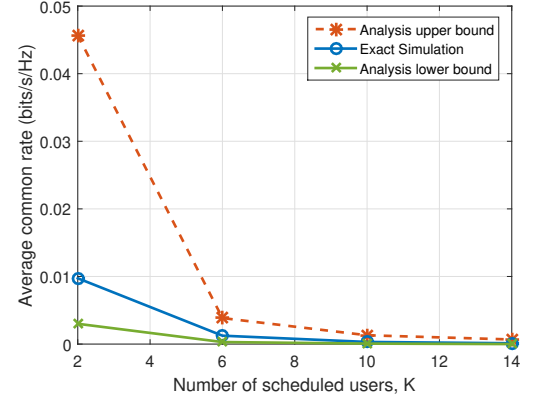}
	}
	\\
	\subfloat[Average RSMA sum rate versus $K$.]{
		\includegraphics[scale=0.5]{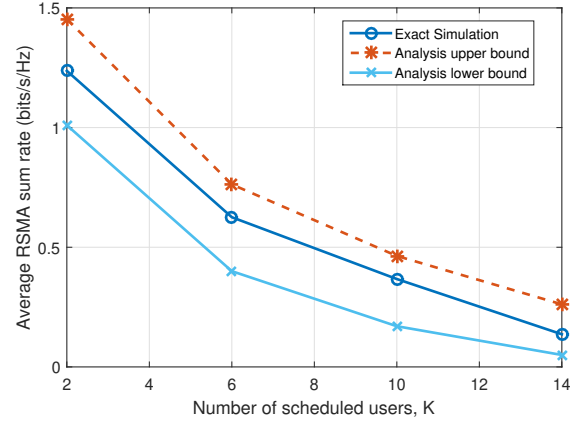}
	}
	\caption{Average RSMA rates versus the number of scheduled users $K$ at $\mathrm{SNR}=10$ dB.}
	\label{fig:rates_vs_K}
\end{figure}
Fig.~\ref{fig:rates_vs_K} investigates the effect of the scheduling size at $\mathrm{SNR}=10$ dB. The close agreement between the Monte Carlo results and the corresponding analytical bounds further supports the validity of the proposed user-scheduling analysis.
As shown in Fig.~\ref{fig:rates_vs_K}(a), the sum of the private rates decreases as $K$ increases. Although a larger $K$ increases the multiplexing order, it simultaneously reduces the ZF beamforming dimension $N_s-K+1$, and under equal private-stream power allocation, reduces the power assigned to each private stream. Moreover, because the private ZF beamformers are designed using the reference-port CSI, residual intra-cell interference may remain at the selected FAS port. Consequently, the loss in per-user rate dominates the multiplexing benefit for the considered network parameters.
Fig.~\ref{fig:rates_vs_K}(b) shows that the common rate decreases even more rapidly with $K$. This behavior is mainly caused by the worst-user decoding requirement of the RSMA common stream. Since the common stream must be decoded by every scheduled user, its achievable rate is limited by the weakest common-stream decoder. Increasing $K$ therefore increases the probability of including a user with an unfavorable serving distance, channel realization, or interference condition, making the common stream increasingly bottlenecked.
Consequently, the RSMA sum rate in Fig.~\ref{fig:rates_vs_K}(c) also decreases with $K$ for this
configuration. These results show that increasing the number of scheduled users does not necessarily maximize the RSMA  throughput, since the potential multiplexing gain can be outweighed by the reduced ZF beamforming gain, lower per-user power, and the common-stream bottleneck. Therefore, $K$ should be carefully optimized, as considered in Section~\ref{sec:resource_allocation}.

\begin{figure}[t]
	\centering
	\includegraphics[scale=0.5]{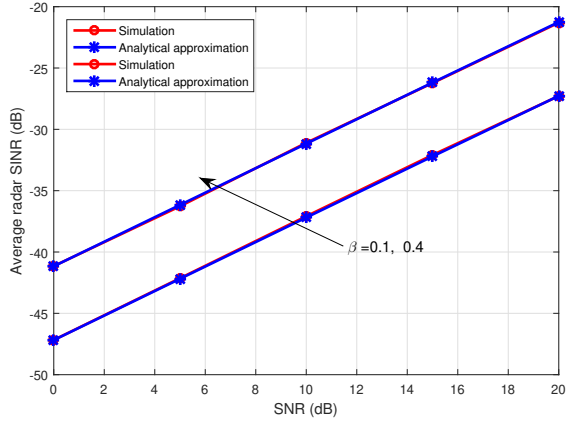}
	\caption{Average sensing SINR versus transmit SNR for different common-stream power allocations.}
	\label{fig:sensing_sinr_snr}
\end{figure}
Fig.~\ref{fig:sensing_sinr_snr} compares the analytical and simulated average sensing SINR for different common-stream power allocations. The analytical approximation closely matches the simulated results over the entire SNR range, confirming the accuracy of the adopted statistical characterization of the common beamforming gain. The average sensing SINR increases with the transmit SNR and improves significantly when a larger fraction of the transmit power is allocated to the common stream, which also serves as the sensing waveform. In particular, increasing the common-stream power from $P_{o,c}=0.1P_o$ to $P_{o,c}=0.4P_o$ yields an approximately $6$-dB gain, consistent with the fourfold increase in sensing power, i.e., $10\log_{10}(0.4/0.1)\simeq6.02$ dB.

\begin{figure}[t]
	\centering
	\includegraphics[scale=0.5]{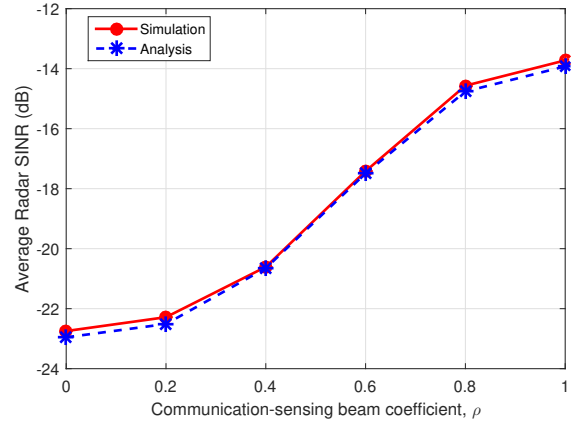}
	\caption{Average RSMA rates versus $\rho$.}
	\label{fig:rho_tradeoff}
\end{figure}
Fig.~\ref{fig:rho_tradeoff} illustrates the average sensing SINR versus the communication-sensing beam tradeoff $\rho$, showing that the average sensing SINR increases substantially with $\rho$. As $\rho$ increases, the target steering-vector component receives a larger weight in the common beamformer, while the communication component $\alpha=\sqrt{1-\rho^2}$ is reduced. The analytical approximation closely follows the Monte Carlo result over the complete range of $\rho$. This observation demonstrates the role of $\rho$ in controlling the communication--sensing operating point and motivates its joint optimization under the sensing QoS constraint.

\begin{figure*}[t]
	\centering
	\subfloat[Optimal RSMA sum rate.]{
		\includegraphics[scale=0.5]{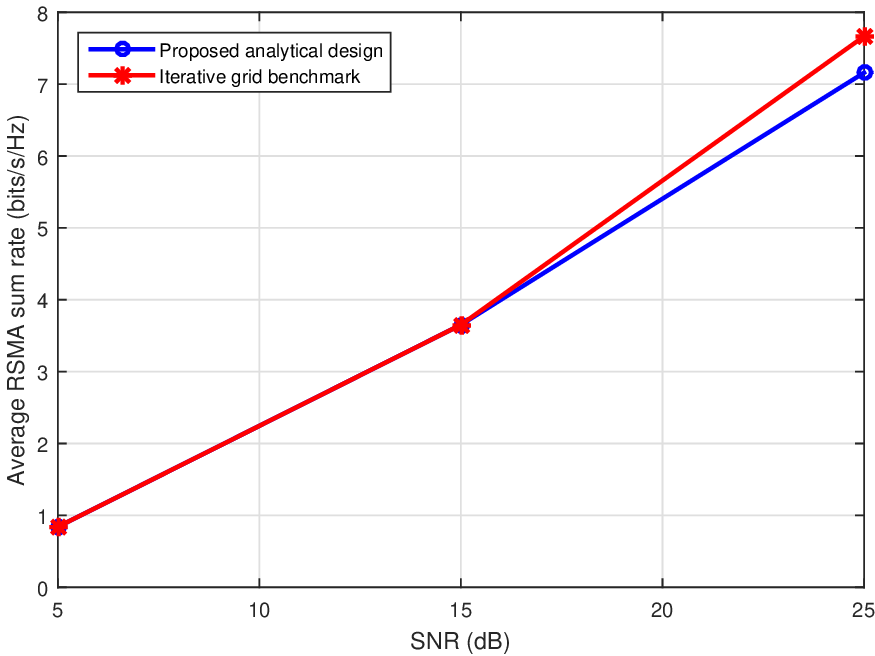}
	}
	\subfloat[Running time vs SNR.]{
		\includegraphics[scale=0.5]{fig7b}
	}
	\\
	\subfloat[Proposed analytical design optimized parameters vs SNR.]{
		\includegraphics[scale=0.5]{fig7c}
	}
	\subfloat[Iterative grid design optimized parameters vs SNR.]{
		\includegraphics[scale=0.5]{fig7d}
	}
	\caption{Performance and optimized parameters versus transmit SNR for $N_s=8$ and $P=8$.}
	\label{fig:optimization_N8_P8}
\end{figure*}
Fig.~\ref{fig:optimization_N8_P8} compares the proposed analytical resource-allocation method with the numerical grid-search benchmark for $N_s=8$ and $P=8$. As shown in
Fig.~\ref{fig:optimization_N8_P8}(a), both methods achieve nearly the same RSMA sum rate over the entire SNR range, with only a small performance loss at high SNR. Thus, the proposed low-complexity analytical method provides a near-optimal solution without performing a multidimensional numerical search. Fig.~\ref{fig:optimization_N8_P8}(b) compares their measured running times. The analytical design requires substantially less than one second over the considered SNR range, whereas the grid-search benchmark requires tens to hundreds of seconds. Hence, the analytical design provides a reduction of several orders of magnitude in measured runtime which demonstrates the effectiveness of the proposed analytical framework for practical implementations while maintaining nearly identical communication performance. 
The optimized parameters are shown in
Figs.~\ref{fig:optimization_N8_P8}(c) and
\ref{fig:optimization_N8_P8}(d). Both methods select $K^\star=2$ throughout the considered SNR range, indicating that serving two users provides the best balance between spatial multiplexing gain and inter-user interference. The common-power coefficient $\beta^\star$ generally increases with SNR, whereas $\rho^\star$ first increases and subsequently decreases at high SNR, indicating
that less steering toward the sensing direction is required once sufficient transmit power becomes available. Most importantly, the analytical design reproduces the same overall parameter trends as the grid-search benchmark. The optimized values of the iterative grid design exhibit nearly identical trends to those obtained by the proposed analytical design, confirming that the proposed optimization accurately predicts the near-optimal resource allocation while requiring substantially lower running time.

\begin{figure*}[t]
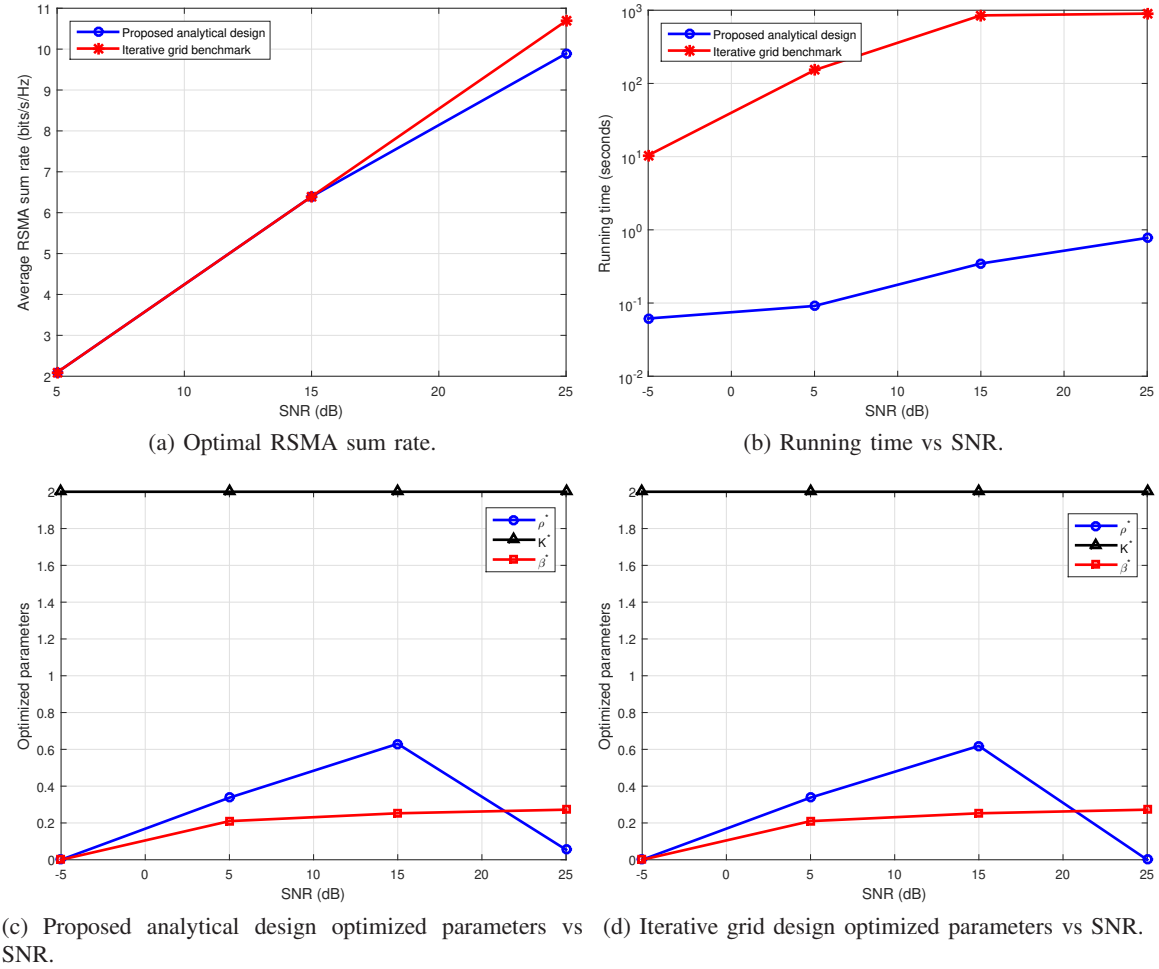

	\centering
	\subfloat[Optimal RSMA sum rate.]{
		\includegraphics[scale=0.5]{fig8a}
	}
	\subfloat[Running time vs SNR.]{
		\includegraphics[scale=0.5]{fig8b}
	}
	\\
	\subfloat[Proposed analytical design optimized parameters vs SNR.]{
		\includegraphics[scale=0.5]{fig8c}
	}
	\subfloat[Iterative grid design optimized parameters vs SNR.]{
		\includegraphics[scale=0.5]{fig8d}
	}
	\caption{Performance and optimized parameters versus transmit SNR for
		$N_s=8$ and $P=16$.}
	\label{fig:optimization_N8_P16}
\end{figure*}
Fig.~\ref{fig:optimization_N8_P16} repeats the comparison for
$N_s=8$ and $P=16$. Increasing the number of FAS ports improves the achievable RSMA sum rate because of the additional port-selection diversity, enabling each user to select a channel realization with a larger desired gain while preserving the same number of RF chains. The proposed analytical design continues to closely track the grid-search benchmark, while retaining a substantially lower runtime. The optimized scheduling size remains $K^\star=2$, indicating that increasing $P$ primarily enhances the quality of the selected user channel rather than the available transmit spatial degrees of freedom. The optimized $\beta^\star$ varies moderately with SNR, whereas $\rho^\star$ increases toward the medium-SNR regime and then decreases at high SNR. The two optimization methods yield very similar parameter trends. Increasing $P$ does not by itself enlarge the dimensionality of the $(\rho,\beta,K)$ search space; rather, it may increase the numerical cost of evaluating each candidate configuration because the FAS selection and its analytical characterization involve a larger number of candidate ports.

\begin{figure*}[t]
	\centering
	\subfloat[Optimal RSMA sum rate.]{
		\includegraphics[scale=0.5]{fig9a}
	}
	\subfloat[Running time vs SNR.]{
		\includegraphics[scale=0.5]{fig9b}
	}
	\\
	\subfloat[Proposed analytical design optimized parameters vs SNR.]{
		\includegraphics[scale=0.5]{fig9c}
	}
	\subfloat[Iterative grid design optimized parameters vs SNR.]{
		\includegraphics[scale=0.5]{fig9d}
	}
	\caption{Performance and optimized parameters versus transmit SNR for
		$N_s=32$ and $P=8$.}
	\label{fig:optimization_N32_P8}
\end{figure*}

Fig.~\ref{fig:optimization_N32_P8} investigates the effect of
increasing the number of transmit antennas to $N_s=32$. The achievable RSMA sum rate increases substantially relative to the $N_s=8$ case because the larger array provides higher beamforming gain, improved interference suppression, and additional spatial degrees of freedom for multiuser transmission. The proposed analytical solution remains close to the grid-search benchmark over the considered SNR range, demonstrating that its accuracy is maintained even for large antenna arrays while maintaining a substantially lower runtime. A particularly important difference appears in the optimized scheduling size. In contrast to Figs.~\ref{fig:optimization_N8_P8} and
\ref{fig:optimization_N8_P16}, where $K^\star=2$, the preferred number of scheduled users increases with SNR for $N_s=32$, reaching $K^\star=7$ at the highest plotted SNR. The larger antenna array therefore allows the system to exploit a higher multiplexing order without incurring the severe ZF beamforming loss observed for smaller arrays. The analytical and grid-search designs select nearly identical
values of $K^\star$, $\beta^\star$, and $\rho^\star$, demonstrating that the proposed low-complexity design captures the principal resource-allocation behavior even for massive MIMO configurations.

\section{Conclusion}

This paper investigated a network-level FAS-assisted RSMA-ISAC framework using stochastic geometry. Closed-form expressions were derived for the average  RSMA sum rate, and radar SINR, accounting for the spatial randomness of BSs, users, and sensing targets. A low-complexity resource allocation framework was then developed to jointly optimize the common power allocation, communication-sensing beamforming parameter, and number of scheduled users under radar sensing constraints. Numerical results showed excellent agreement between the analytical results and Monte-Carlo simulations, while the proposed optimization achieved performance close to exhaustive grid search with substantially lower complexity. The results further demonstrated that the proposed beamforming strategy also effectively balances communication and sensing performance across different operating conditions.

\section*{Appendix A}

Let $X=d_{k,o}^{-\alpha_c}P_{o,k}G_{k,max}^p$
and $Y=Y_{intra}+Y_{inter}+\sigma_{k}^{2}$
with $Y_{intra}=\stackrel[i\neq k]{K}{\sum}d_{k,o}^{-\alpha_c}P_{o,i}\left|\mathbf{h}_{k,p,o}\frac{\mathbf{\tilde{w}}_{o,i}^{\mathrm{zf}}}{\left\Vert \mathbf{\tilde{w}}_{o,i}^{\mathrm{zf}}\right\Vert }\right|^{2}$and
$Y_{inter}=\underset{l\in\Phi_{b}\setminus o}{\sum}d_{k,l}^{-\alpha_c}\left\Vert \mathbf{h}_{k,p,l}\mathbf{W}_{l}\right\Vert ^{2}$. Now the average of the private part can be calculated by {[}useful
lemma{]}
\begin{equation}
	\begin{aligned}
		\mathcal{E}\left\{ \log\left(1+\gamma_{k}^{p}\right)\right\} =\stackrel[0]{\infty}{\int}\frac{1}{z}\left(1-\mathcal{E}\left\{ e^{-zX}\right\} \right)\mathcal{E}\left\{ e^{-zY}\right\} dz.
	\end{aligned}
\end{equation}

The exact average rate is intractable; hence, we derive upper and lower bounds on the average private rate. The Laplace transform
of $X$, conditioned on $d_{k,o}=r_{o}$, is $
\mathcal{E}\left\{ e^{-zX}\right\} =\mathcal{E}\left\{ e^{-zr_{o}^{-\alpha_c}P_{o,k}G_{k,max}^p}\right\}$. Since $G_{k,max}^p=\left|\mathbf{h}_{k,p,o}\frac{\mathbf{\tilde{w}}_{o,k}^{\mathrm{zf}}}{\left\Vert \mathbf{\tilde{w}}_{o,k}^{\mathrm{zf}}\right\Vert }\right|^{2}\sim\Gamma\left(N_{s}-K+1,1\right)$ with
CDF $
F_{G_{p}}\left(x\right)=1-e^{-x}\stackrel[n=0]{m-1}{\sum}\frac{x^{n}}{n!}$, where $m=N_{s}-K+1$. Thus, the CDF of $G_{k,max}^p$ can be upper bounded
as $F_{G_{k,max}^p}\left(x\right)=\left[1-e^{-x}\stackrel[n=0]{m-1}{\sum}\frac{x^{n}}{n!}\right]^{P}$. Thus, the Laplace transform of $X$, conditioned on $d_{k,o}=r_{o}$, can be calculated by
\begin{equation}
	\begin{aligned}
		\mathcal{E}\!\left(e^{-zX}\mid r_o\right)
		&=z r_o^{-\alpha_c}P_{o,k}
		\int_{0}^{\infty}
		e^{-z r_o^{-\alpha_c}P_{o,k}x}  \\
		&\quad\times
		\left[
		1-e^{-x}\sum_{n=0}^{m-1}\frac{x^n}{n!}
		\right]^P dx .
	\end{aligned}
\end{equation}
By applying the binomial expansion,

\begin{equation}
	\begin{aligned}
		\left[1-e^{-x}\stackrel[n=0]{m-1}{\sum}\frac{x^{n}}{n!}\right]^{P}
		&=\stackrel[r=0]{P}{\sum}\left(-1\right)^{r}\left(\begin{array}{c}
			P\\
			r
		\end{array}\right)e^{-rx}\\
		&\quad\times\left(\stackrel[n=0]{m-1}{\sum}\frac{x^{n}}{n!}\right)^{r}
	\end{aligned}
\end{equation}
also $
\left(\stackrel[n=0]{m-1}{\sum}\frac{x^{n}}{n!}\right)^{r}=\stackrel[q=0]{r\left(m-1\right)}{\sum}c_{r,q}x^{q}$
where $c_{r,q}=\underset{\underset{0\leq k_{i}\leq m-1}{k_{1}+..+k_{r}=q}}{\sum}\frac{1}{k_{1}!k_{2}!..k_{r}!}$.
Thus,
\begin{equation}
	\begin{aligned}
		\mathcal{E}\!\left(e^{-zX}\mid r_{o}\right)
		&=\frac{z r_o^{-\alpha_c}P_{o,k}}{\Gamma(m)}
		\sum_{r=0}^{P}(-1)^r\binom{P}{r} \\
		&\quad\times
		\sum_{q=0}^{r(m-1)}
		\frac{c_{r,q}\Gamma(m+q)}
		{\bigl(r_o^{-\alpha_c}P_{o,k}z+r\bigr)^{m+q}} .
	\end{aligned}
\end{equation}

Finally, averaging over the $k$-th nearest-user distance with PDF
$
f_{d_{k,o}}\left(r_{o}\right)=\frac{2\left(\pi\lambda_{u}\right)^{k}}{\left(k-1\right)!}r_{o}^{2k-1}e^{-\pi\lambda_{u}r_{o}^{2}}
$, gives the Laplace transform of $X$ as
\begin{equation}
	\begin{aligned}
		\mathcal{L}_{X}(z)=\mathcal{E}\left\{ e^{-zX}\right\}
		=\frac{2(\pi\lambda_u)^k zP_{o,k}}
		{(k-1)!\Gamma(m)}
		\sum_{r=0}^{P}(-1)^r\binom{P}{r} \\
		\times
		\sum_{q=0}^{r(m-1)}c_{r,q}\Gamma(m+q) 
		\times
		\int_{0}^{\infty}
		\frac{r_o^{2k-1-\alpha_c}e^{-\pi\lambda_u r_o^2}}
		{\bigl(r_o^{-\alpha_c}P_{o,k}z+r\bigr)^{m+q}}
		\,dr_o .
	\end{aligned}
\end{equation}

Conditioned on $d_{k,o}=r_{o}$,
$\mathcal{L}_{Y}(z)=\mathcal{E}\!\left\{e^{-zY}\right\}=\mathcal{E}\!\left[
e^{-z(Y_{\mathrm{intra}}
	+Y_{\mathrm{inter}}
	+\sigma_k^2)}
\right]=\mathcal{L}_{Y_{intra}}\left(z\right)\mathcal{L}_{Y_{inter}}\left(z\right)e^{-z \sigma_{k}^{2}}$. Since $\stackrel[i\neq k]{K}{\sum}\left|\mathbf{h}_{k,p,o}\frac{\mathbf{\tilde{w}}_{o,i}^{\mathrm{zf}}}{\left\Vert \mathbf{\tilde{w}}_{o,i}^{\mathrm{zf}}\right\Vert }\right|^{2}\sim\Gamma\left(K-1,1\right)$, we can write 
\begin{equation}
	\mathcal{L}_{Y_{intra}}\left(z/r_{o}\right)=\left(1+r_{o}^{-\alpha_c}P_{o,i}z\right)^{1-K}.
\end{equation}

Assuming equal power for all users, $P_{o,i}=P_{o,k}$ and veraging over the $k$-th nearest-user distance gives
\begin{equation}
	\begin{aligned}
		\mathcal{L}_{Y_{\mathrm{intra}}}(z)
		&=\frac{2(\pi\lambda_u)^k}{(k-1)!}
		\int_{0}^{\infty}
		\bigl(1+r_o^{-\alpha_c}P_{o,k}z\bigr)^{1-K} \\
		&\quad\times
		r_o^{2k-1}e^{-\pi\lambda_u r_o^2}\,dr_o .
	\end{aligned}
\end{equation}

Similarly, since $
\left\Vert \mathbf{h}_{k,p,l}\mathbf{W}_{l}\right\Vert ^{2}\sim\Gamma\left(K,1\right)$
we can find, $
\mathcal{E}\left\{ e^{-z\left(d_{k,l}^{-\alpha_c}\left\Vert \mathbf{h}_{k,p,l}\mathbf{W}_{l}\right\Vert ^{2}\right)}\right\} =\left(1+d_{k,l}^{-\alpha_c}z\right)^{-K}
$. Thus,
\begin{equation}
	\begin{aligned}
		\mathcal{L}_{Y_{\mathrm{inter}}}(z)
		&=\mathcal{E}\!\left[
		\prod_{l\in\Phi_b\setminus\{o\}}
		\bigl(1+d_{k,l}^{-\alpha_c}z\bigr)^{-K}
		\right].
	\end{aligned}
\end{equation}
Using PGFL we can obtain,
\begin{equation}
	\begin{aligned}
		\mathcal{L}_{Y_{\mathrm{inter}}}(z)
		&=e^{-2\pi\lambda_{b}\stackrel[r_{l}]{\infty}{\int}\left(1-\left(1+d_{k,l}^{-\alpha_c}z\right)^{-K}\right)d_{k,l}d_{d_{k,l}}}
	\end{aligned}
\end{equation}

Finally, substituting $\mathcal{L}_{X}(z)$ and $\mathcal{L}_{Y}(z)$ into (1) and applying the Gauss?Laguerre quadrature yields the average rate in Theorem 1.

\section*{Appendix B}
By using Jensen's inequality, we can write
\begin{equation}
	R_{k}^{p}\geq\frac{1}{\ln2}\ln\left(1+\frac{e^{\mathcal{E}\ln\left(X\right)}}{\mathcal{E}\left(Y\right)}\right)
\end{equation}
First, 
\begin{equation}
	\begin{aligned}
		\mathcal{E}\left\{ \ln\left(X\mid r_o\right)\right\} =\ln P_{o,k}-\alpha_c\ln r_o+\mathcal{E}\ln\left(G_{k,max}^p\right)\\
		=\ln P_{o,k}-\alpha_c\ln r_o+\psi\left(N_{s}-K+1\right)
	\end{aligned}
\end{equation}
where $\psi\left(.\right)$ is the digamma function. Also,
\begin{equation}
	\mathcal{E}\left\{ \ln\left(Y\mid r_o\right)\right\} =P_{o,k}r_o^{-\alpha_c}\left(K-1\right)+\frac{2\pi\lambda_{b}K}{\alpha_c-2}r_o^{2-\alpha_c}+\sigma_{k}^{2}
\end{equation}
Therefore, a valid lower bound is
\begin{equation}
	\begin{aligned}
		R_{k,\mathrm{Lb}}^{p}
		&=\frac{1}{\ln 2}
		\int_{0}^{\infty}
		\ln\!\left(1+\frac{N_p(r_o)}{D_p(r_o)}\right)
		\times f_{d_{k,o}}(r_o)\,dr_o 
	\end{aligned}
\end{equation}
where 
$
N_p(r)
=P_{o,k}r_o^{-\alpha_c}
+e^{\psi(N_s-K+1)},\\
D_p(r)
=P_{o,k}r_o^{-\alpha_c}(K-1)
+\frac{2\pi\lambda_bK}{\alpha_c-2}r^{2-\alpha_c}
+\sigma_k^2, and\\    f_{d_{k,o}}\left(r_{o}\right)=\frac{2\left(\pi\lambda_{u}\right)^{k}}{\left(k-1\right)!}r_{o}^{2k-1}e^{-\pi\lambda_{u}r_{o}^{2}}.$

Finally, implementing the Lagrange polynomial, we can get the average rate as in Theorem 2.

\section*{Appendix C}

Let $\gamma_{k}^{c}=\frac{X_{c}}{Y_{c}}$, with $ X_{c}=d_{k,o}^{-\alpha_c}P_{o,c}\left|\mathbf{h}_{k,p,o}\mathbf{w}_{o,c}\right|^{2}$ and
$Y_c=Y_{c,\mathrm{intra}}+Y_{c,\mathrm{inter}}+\sigma_k^2$, where
$Y_{c,\mathrm{intra}}=\sum_{i=1}^{K}
d_{k,o}^{-\alpha_c}P_{o,i}
\left|
\mathbf{h}_{k,p,o}
\frac{\widetilde{\mathbf w}_{o,i}^{\mathrm{zf}}}
{\|\widetilde{\mathbf w}_{o,i}^{\mathrm{zf}}\|}
\right|^2,$ and 
$Y_{c,\mathrm{inter}}=\sum_{l\in\Phi_b\setminus\{o\}}
d_{k,l}^{-\alpha_c}\|\mathbf h_{k,p,l}\mathbf W_l\|^2$. Since $\mathbf{w}_{o,c}$ is constructed from the user channels and sensing steering vector, the gain $G_{c}=\left|\mathbf{h}_{k,p,o}\mathbf{w}_{o,c}\right|^{2}$, can be approximated using moment matching as $G_{c}\sim\Gamma\left(m_{c},\theta_{c}\right)$
where $m_{c}=\frac{\left(\mathcal{E}\left(G_{c}\right)\right)^{2}}{\mathrm{Var}\left(G_{c}\right)}$
and $\theta_{c}=\frac{\mathrm{Var}\left(G_{c}\right)}{\mathcal{E}\left(G_{c}\right)}$.
A simple usable approximations are $\mathcal{E}\left(G_{c}\right)\thickapprox1+\frac{N_{s}\left|\alpha\right|^{2}}{\stackrel[i=1]{K}{\sum}\left|\alpha\right|^{2}+\rho^{2}}, \mathrm{Var}\left(G_{c}\right)\thickapprox1+\frac{2N_{s}\left|\alpha\right|^{2}}{\stackrel[i=1]{K}{\sum}\left|\alpha\right|^{2}+\rho^{2}}$.

Conditioned on $d_{k,o}=r_{o}$, we can find 
\begin{equation}
	\mathcal{L}_{X_{c}}\left(z\mid r_{o}\right)=\left(1+\theta_{c}zr_{o}^{-\alpha_c}P_{o,c}\right)^{-m_{c}}.
\end{equation}

Averaging over the $k$-th nearest-user distance gives
\begin{equation}
	\begin{aligned}
		\mathcal{L}_{X_c}(z)
		&=\frac{2(\pi\lambda_u)^k}{(k-1)!}
		\int_{0}^{\infty}
		\bigl(1+\theta_c z r_o^{-\alpha_c}P_{o,c}\bigr)^{-m_c}\\
		&\quad\times
		r_o^{2k-1}e^{-\pi\lambda_u r_o^2}\,dr_o .
	\end{aligned}
\end{equation}

Since $\stackrel[i=1]{K}{\sum}\left|\mathbf{h}_{k,p,o}\frac{\mathbf{\tilde{w}}_{o,i}^{\mathrm{zf}}}{\left\Vert \mathbf{\tilde{w}}_{o,i}^{\mathrm{zf}}\right\Vert }\right|^{2}\sim\Gamma\left(K,1\right)$, we can get

\begin{equation*}
	\begin{aligned}
		\mathcal{L}_{Y_{c,\mathrm{intra}}}(z)
		&=\frac{2(\pi\lambda_u)^k}{(k-1)!}
		\int_{0}^{\infty}
		\bigl(1+r_o^{-\alpha_c}P_{o,k}z\bigr)^{-K}\\
		&\quad\times
		r_o^{2k-1}e^{-\pi\lambda_u r_o^2}\,dr_o 
	\end{aligned}
\end{equation*}
and $\mathcal{L}_{Y_{c,inter}}\left(z\right)=e^{-2\lambda_{b}\pi\stackrel[r_{l}]{\infty}{\int}\left(1-\left(1+d_{k,l}^{-\alpha_c}z\right)^{-K}\right)d_{k,l}d_{d_{k,l}}}$

The average common decoding rate at user $k$ can be obtained by substituting $\mathcal{L}_{X_{c}}\left(z\right)$ and $\mathcal{L}_{Y_{c}}\left(z\right)=\mathcal{L}_{Y_{c,intra}}\left(z\right)\mathcal{L}_{Y_{c,inter}}\left(z\right)e^{-z\sigma_{k}^{2}}$ in (1). Using Laguerre polynomials, we can get the average rate as in Theorem 3.

\section*{Appendix D}
For common decoding at user $k$, using Jensen's inequality, we can get

\begin{equation}
	R_{k}^{c}\geq\frac{1}{\ln2}\ln\left(1+\frac{e^{\mathcal{E}\ln\left(X_{c}\right)}}{\mathcal{E}\left(Y_{c}\right)}\right)
\end{equation}
where $X_{c}=d_{k,o}^{-\alpha_c}P_{o,c}G_{c}$ and
$Y_{c}=Y_{c,intra}+Y_{c,inter}+\sigma_{k}^{2}$. The average of the first term is 
\begin{equation}
	\mathcal{E}\left[\ln\left(X_{c}\right)\mid r_o\right]=\ln P_{o,c}-\alpha_c\ln r_{o}+\mathcal{E}\ln\left(G_{c}\right)
\end{equation}

For $G_{c}\simeq\Gamma\left(m_{c},\theta_{c}\right)$, $
\mathcal{E}\ln\left(G_{c}\right)=\psi\left(m_{c}\right)+\ln\theta_{c}$. Also,
\begin{equation}
	\mathcal{E}\left[\ln\left(Y_{c}\right)\mid r_o\right]=P_{o,k}r_{o}^{-\alpha_c}K+\frac{2\pi\lambda_{b}K}{\alpha_c-2}r_{o}^{2-\alpha_c}+\sigma_{k}^{2}
\end{equation}

Therefore, a valid lower bound is
\begin{equation}
	\begin{aligned}
		R_{k,\mathrm{Lb}}^{c}
		&=\frac{1}{\ln 2}
		\int_{0}^{\infty}
		\ln\!\left(1+\frac{N_c(r_o)}{D_c(r_o)}\right) \times f_{d_{k,o}}(r_o)\,dr_o .
	\end{aligned}
\end{equation}
where $
N_c(r_o)=P_{o,c} r_o^{-\alpha_c}\theta_c e^{\psi(m_c)}, \text{ and } 
D_c(r_o)=KP_{o,i} r_o^{-\alpha_c}
+\frac{2\pi\lambda_bK}{\alpha_c-2}r_o^{2-\alpha_c}
+\sigma_k^2 .$ 
Using Laguree polynomial we can get the average rate as in Theorem 4.

\section*{Appendix E}

Lets define, $\mathcal{E}\left[\gamma_{s}\right]=\mathcal{E}\left[\frac{X_{s}}{Y_{s}+\sigma_{b}^{2}}\right]$, where 
$X_{s}=P_{o,c}\beta_{m,o}\zeta_{m,o}A_{s}$ and $
A_{s}=\left|\mathbf{w}_{r}^{H}\mathbf{a}_{r}\left(\theta_{m}\right)\right|^{2}\left|\mathbf{a}_{t}^{H}\left(\theta_{m}\right)\mathbf{w}_{o,c}\right|^{2}$. Now, $ \mathcal{E}\left[X_{s}\right]=\zeta_{m,o}P_{o,c}\mathcal{E}\left[\beta_{m,o}\right]\mathcal{E}\left[A_{s}\right]$ where
$
\mathcal{E}\left[A_{s}\right]=\left|\mathbf{w}_{r}^{H}\mathbf{a}_{r}\left(\theta_{m}\right)\right|^{2}\mathcal{E}\left[\left|\mathbf{a}_{t}^{H}\left(\theta_{m}\right)\mathbf{w}_{o,c}\right|^{2}\right]
$. 
The common precoder is $\mathbf{w}_{o,c}=\frac{\stackrel[i=1]{K}{\sum}\alpha\mathbf{h}_{i,1,o}+\rho\mathbf{a}_{t}\left(\theta_{m}\right)}{\left\Vert \stackrel[i=1]{K}{\sum}\alpha\mathbf{h}_{i,1,o}+\rho\mathbf{a}_{t}\left(\theta_{m}\right)\right\Vert }$.
Let $\mathbf{u}=\stackrel[i=1]{K}{\sum}\alpha\mathbf{h}_{i,1,o}$. Since $
\mathbf{u}\sim \mathcal{CN}\left(0,A_{\alpha}\mathbf{I}\right)$, where $A_{\alpha}={K}\left|\alpha\right|^{2}$.
Now we can write, $ \mathbf{v}=\mathbf{u}+\rho\mathbf{a}_{t}$, and $
\mathbf{w}_{o,c}=\frac{\mathbf{v}}{\left\Vert \mathbf{v}\right\Vert }$. Thus, $
\mathcal{E}\left[\left|\mathbf{a}_{t}^{H}\left(\theta_{m}\right)\mathbf{w}_{o,c}\right|^{2}\right]=\mathcal{E}\left[\left|\frac{\mathbf{a}_{t}^{H}\mathbf{v}}{\left\Vert \mathbf{v}\right\Vert }\right|^{2}\right]$. The ratio of means approximation is $\mathcal{E}\left\{ \left|\mathbf{a}_{t}^{H}\left(\theta_{m}\right)\mathbf{w}_{o,c}\right|^{2}\right\} =\frac{\mathcal{E}\left|\mathbf{a}_{t}^{H}\mathbf{v}\right|^{2}}{\mathcal{E}\left\Vert \mathbf{v}\right\Vert ^{2}}
$ where
$
\mathcal{E}\left|\mathbf{a}_{t}^{H}\left(\theta_{m}\right)\mathbf{v}\right|^{2}=A_{\alpha}\left\Vert \mathbf{a}_{t}\left(\theta_{m}\right)\right\Vert ^{2}+\rho^{2}\left\Vert \mathbf{a}_{t}\left(\theta_{m}\right)\right\Vert ^{4},$ and $\mathcal{E}\left\Vert \mathbf{v}\right\Vert ^{2}=A_{\alpha}N_{s}+\rho^{2}\left\Vert \mathbf{a}_{t}\left(\theta_{m}\right)\right\Vert ^{2}$, and $\left\Vert \mathbf{a}_{t}\left(\theta_{m}\right)\right\Vert ^{2}=N_{s}$. Thus, $\mathcal{E}\left[\left|\mathbf{a}_{t}^{H}\left(\theta_{m}\right)\mathbf{w}_{o,c}\right|^{2}\right]\approx\frac{A_{\alpha}+\rho^{2}N_{s}}{A_{\alpha}+\rho^{2}}$, where $A_{\alpha}=K\left|\alpha\right|^{2}=K(1-\rho^{2})$. Thus
$
\mathcal{E}\left[A_{s}\right]=\left|\mathbf{w}_{r}^{H}\mathbf{a}_{r}\left(\theta_{m}\right)\right|^{2}\frac{A_{\alpha}+\rho^{2}N_{s}}{A_{\alpha}+\rho^{2}}$. Also,  we found, $ \mathcal{E}\left[Y_{s}\left|d_{t}=r\right.\right]=\mathcal{E}\left[\sum_{l\in\Phi_{b}\setminus\{o\}}
		d_{l,o}^{-\alpha_s}
		{Z_l}\right]
$ where ${Z_l}=\left|
		\mathbf{w}_{r}^{H}
		\mathbf{H}_{l,o}
		\mathbf{x}_{l}
		\right|^{2}$. 
Assuming identical interfering power ${P_l}$, independence between $Z_l$ and the BS locations, and an exclusion region $d_l \geq r$, Campbell's theorem gives     
        $ \mathcal{E}\left[Y_{s}\left|d_{t}=r\right.\right]=\frac{2\pi\lambda_{b} P_l K}{\alpha_{s}-2}r^{2-\alpha_{s}}
$. Thus,
$
\mathcal{E}[\gamma_s]
=\frac{(\alpha_s-2)\Xi_s}
{2\pi\lambda_b P_l K}\times \mathcal{E}_{d_{t}}\left[d_{t}^{-(2+\alpha_{s})}\right],$
where $\quad \Xi_s
=
\zeta_{m,o}P_{o,c}
\left|
\mathbf w_r^H\mathbf a_r(\theta_m)
\right|^2
\frac{A_\alpha+\rho^2N_s}
{A_\alpha+\rho^2}.$ For nearest target distance with altitude $h_{t}$ the PDF is $
f_{d_{t}}\left(r\right)=2r\pi\lambda_{r}e^{-\pi\lambda_{r}\left(r^{2}-h_{t}^{2}\right)}$. Thus, $
\mathcal{E}_{d_{t}}\left[d_{t}^{-2-\alpha_{s}}\right]=\pi\lambda_{r}(\pi\lambda_{r})^{\alpha_{s}/2}e^{\pi\lambda_{r}h_{t}^{2}} \Gamma\left(-\alpha_{s}/2,\pi\lambda_{r}h_{t}^{2}\right)
$. Finally,   the average SINR can be expressed as in Theorem 5.

\bibliographystyle{IEEEtran}
\bibliography{ref}

\end{document}